\documentclass[12pt]{article}

\usepackage{times}

\usepackage{algorithm}

\usepackage[noend]{algorithmic}

\usepackage{amssymb,amsmath,amsfonts,amsthm}

\usepackage{epsfig,subfigure}

\usepackage{latexsym,enumerate}

\usepackage{fullpage}

\usepackage{verbatim}

\newcommand{\E}{\ensuremath{\mathcal{E}}}

\newcommand{\V}{\ensuremath{\mathcal{V}}}

\newcommand{\Species}{\ensuremath{S}}

\newcommand{\Leaves}{\ensuremath{\mathcal{L}}}
\newcommand{\FullRef}{\ensuremath{\mathcal{F}}}
\newcommand{\DiffRes}{\ensuremath{\mathcal{D}}}
\newcommand{\SameRes}{\ensuremath{\mathcal{S}}}
\newcommand{\UnRes}{\ensuremath{\mathcal{U}}}
\newcommand{\Res}{\ensuremath{\mathcal{R}}}
\newcommand{\TripletMap}{\ensuremath{\mathcal{M}}}

\newcommand{\Prof}{\ensuremath{\mathcal{P}}}

\newcommand{\Phyls}{\ensuremath{P}}
\newcommand{\RPhyls}{\ensuremath{RP}}

\DeclareMathOperator{\pa}{\textit{pa}}%
\DeclareMathOperator{\ch}{\textit{Ch}}
\DeclareMathOperator{\rt}{\textit{rt}}%
\DeclareMathOperator{\adj}{\textit{adj}}

\def\argmin{\mathop{\rm arg\,min}}

\newtheorem{theorem}{Theorem}[section]

\newtheorem{lemma}{Lemma}[section]
\newtheorem{proposition}{Proposition}[section]

\begin{document}

\setcounter{footnote}{1}

\title{Comparing and Aggregating Partially Resolved Trees\thanks{An extended abstract of this paper was presented at the 8th Latin American Symposium on Theoretical Informatics, B\'uzios, Brazil.}}

\author{Mukul S. Bansal\thanks{Department of Computer Science, Iowa State
University, Ames, IA 50011, USA.  Email:
\texttt{\{bansal, jdong, fernande\}@iastate.edu}.  The authors were
supported in part by National Science Foundation AToL grants
DEB-0334832 and DEB-0829674.} \and Jianrong Dong\footnotemark[\value{footnote}]
\and David Fern\'andez-Baca\footnotemark[\value{footnote}] }

\date{}

\maketitle

%\vfill

\begin{abstract}
We define, analyze, and give efficient algorithms for two kinds of
distance measures for rooted and unrooted phylogenies.  For rooted
trees, our measures are based on the topologies the input trees
induce on \emph{triplets}; that is, on three-element subsets of
the set of species.  For unrooted trees, the measures are based on
\emph{quartets} (four-element subsets). Triplet and quartet-based
distances provide a robust and fine-grained measure of the
similarities between trees.  The distinguishing feature of our
distance measures relative to traditional quartet and triplet
distances is their ability to deal cleanly with the presence of
unresolved nodes, also called polytomies. For rooted trees, these
are nodes with more than two children; for unrooted trees, they
are nodes of degree greater than three.

Our first class of measures are parametric distances, where there
is a parameter that weighs the difference between an unresolved
triplet/quartet topology and a resolved one.  Our second class of
measures are based on Hausdorff distance.  Each tree is viewed as
a set of all possible ways in which the tree could be refined to
eliminate unresolved nodes.  The distance between the original
(unresolved) trees is then taken to be the Hausdorff distance
between the associated sets of fully resolved trees, where the
distance between trees in the sets is the triplet or quartet
distance, as appropriate.


\end{abstract}

%\vfill \thispagestyle{empty}
%
%\pagebreak
%
%\setcounter{page}{1}

\section{Introduction}

Evolutionary trees, also known as phylogenetic trees or
phylogenies, represent the evolutionary history of sets of
species.  Such trees have uniquely labeled leaves, corresponding
to the species, and unlabeled internal nodes, representing
hypothetical ancestors. The trees can be either rooted, if the
evolutionary origin is known, or unrooted, otherwise.

This paper addresses two related questions: (1) How does one
measure how close two evolutionary trees are to each other? (2)
How does one combine or \emph{aggregate} the phylogenetic
information from conflicting trees into a single \emph{consensus
tree}?  Among the motivations for the first question is the growth
of phylogenetic databases, such as TreeBase \cite{TreeBASE}, with
the attendant need for sophisticated querying mechanisms and for
means to assess the quality of answers to queries.  The second
question arises from the fact that phylogenetic analyses --- e.g.,
by parsimony \cite{Felsenstein03} --- typically produce multiple
evolutionary trees (often in the thousands) for the same set of
species. Another motivation is the ongoing effort to assemble the
tree of life by piecing together phylogenies for subsets of
species \cite{DriskellAneBurleighMcMahonOMearaSanderson04}.

We address the above questions by defining appropriate
\emph{distance measures} between trees.  While several such
measures have been proposed before (see below), ours provide a
feature that previous ones do not: The ability to deal cleanly
with the presence of \emph{unresolved} nodes, also called
\emph{polytomies}.   For rooted trees these are nodes with more
than two children; for unrooted trees, they are nodes of degree
greater than three.  Polytomies cannot simply be ignored, since
they arise naturally in phylogenetic analysis. Furthermore, they
must be treated with care: A node may be unresolved because it
truly must be so or because there is not enough evidence to break
it up into resolved nodes
--- that is, the polytomies are either ``hard'' or ``soft''
\cite{Maddison89}.

\paragraph{Our contributions.}
We define and analyze two kinds of distance measures for
phylogenies.  For rooted trees, our measures are based on the
topologies the input trees induce on \emph{triplets}; that is, on
three-element subsets of the set of species.  For unrooted trees,
the measures are based on \emph{quartets} (four-element subsets).
Our approach is motivated by the observation that triplet and
quartet topologies are the basic building blocks of rooted and
unrooted trees, in the sense that they are the smallest
topological units that completely identify a phylogenetic tree
\cite{SempleSteel03}. Triplet and quartet-based distances thus
provide a robust and fine-grained measure of the differences and
similarities between trees\footnote{Biologically-inspired
arguments in favor of triplet-based measures can be found in
\cite{CottonSlaterWilkinson06}.}.  In contrast with traditional
quartet and triplet distances, our two classes of distance
measures deal cleanly with the presence of unresolved nodes.  Each
of them does so in a different way.

The first kind of measures we propose are \emph{parametric
distances}:  Given a triplet (quartet) $X$, we compare the
topologies that each of the two input trees induces on $X$.  If
they are identical, the contribution of $X$ to the distance is
zero.  If both topologies are fully resolved but different, then
the contribution is one. Otherwise, the topology is resolved in
one of the trees, but not the other. In this case, $X$ contributes
$p$ to the distance, where $p$ is a real number between $0$ and
$1$. Parameter $p$ allows one to make a smooth transition between
hard and soft views of polytomy. At one extreme, if $p = 1$, an
unresolved topology is viewed as different from a fully resolved
one.  At the other, when $p = 0$, unresolved topologies are viewed
as identical to resolved ones.  Intermediate values of $p$ allow
one to adjust for the degree of certainty one has about a
polytomy.

The second kind of measures proposed here are based on viewing
each tree as a set of all possible fully resolved trees that can
be obtained from it by refining its unresolved nodes. The distance
between two trees is defined as the Hausdorff distance between the
corresponding sets\footnote{Informally, two sets $A$ and $B$ are
at Hausdorff distance $\tau$ of each other if each element of $A$
is within distance $\tau$ of $B$ and vice-versa. For a formal
definition, see Section~\ref{sec:distances}.}, where the distance
between trees in the sets is the triplet or quartet distance, as
appropriate.
%Note that the issue of resolved versus unresolved
%triplets or quartets does not arise for trees in the sets.

After defining our distance measures, we proceed to study their
mathematical and algorithmic properties. We obtain exact and
asymptotic bounds on expected values of parametric triplet
distance and parametric quartet distance.
% under the assumption
%that trees are chosen uniformly at random with replacement.
We also study for which values of $p$, parametric triplet and
quartet distances are metrics, \emph{near-metrics} (in the sense
of \cite{FaginKumarMahdianSivakumarVee06}), or non-metrics.

Aside from the mathematical elegance that metrics and near-metrics
bring to tree comparison, there are also algorithmic benefits. We
formulate phylogeny aggregation as a \emph{median} problem, in
which the objective is to find a consensus tree whose total
distance to the given trees is minimized. We do not know whether
finding the median tree relative to parametric (triplet or
quartet) distance is NP-hard, but conjecture that it is.  This is
suggested by the NP-completeness of the \emph{maximum triplet
compatibility problem}\footnote{The input to this problem consists
of a set of trees, each of which has three leaves; the leaf sets
of these trees may not be identical. The question is to find the
largest subset of these triplet trees such that all of the trees
are consistent with a single tree $T$ whose leaf set is the union
of the leaves of the input triplet trees.} \cite{Bryant97}.
However, by the results mentioned above and well-known facts about
the median problem \cite{Vazirani01}, there are simple
constant-factor approximation algorithms for the aggregation of
rooted and unrooted trees relative to parametric distance: Simply
return the input tree with minimum distance to the remaining input
trees. We show that there are values of $p$ for which parametric
distance is a metric, but the median tree may not be fully
resolved even if all the input trees are. However, beyond a
threshold, the median tree is guaranteed to be fully resolved if
the input trees are fully resolved.

A natural problem is whether Hausdorff triplet (quartet) distance
between two trees can be computed in polynomial time. We suspect
that computing Hausdorff triplet (quartet) distance is NP-hard.
However, even if this were so, we show that one can partially
circumvent the issue by proving that, under a certain density
assumption, Hausdorff distance is within a constant factor of
parametric distance --- that is, the measures are
\emph{equivalent} in the sense of
\cite{FaginKumarMahdianSivakumarVee06}.

Finally, we present a $O(n^2)$-time algorithm to compute
parametric triplet distance and a $O(n^2)$ 2-approximate algorithm
for parametric quartet distance. To our knowledge, there was no
previous algorithm for computing the parametric triplet distance
between two rooted trees, other than by enumerating all
$\Theta(n^3)$ triplets. Two algorithms exist that can be directly
applied to compute the parametric quartet distance (see also
\cite{BryantTsangKearneyLi00}).  One runs in time $O(n^2
\min\{d_1,d_2\})$, where, for $i \in \{1, 2\}$, $d_i$ is the
maximum degree of a node in $T_i$
\cite{ChristiansenMailundPedersenRandersStigStissing06}; the other
takes $O(d^9 n \log n)$ time, where $d$ is the maximum degree of a
node in $T_1$ and $T_2$ \cite{Stissing:et:al:07}.\footnote{Note
that the presence of unresolved nodes seems to complicate distance
computation. Indeed, the quartet distance between a pair of
\emph{fully resolved} unrooted trees can be obtained in $O(n \log
n)$ time \cite{BrodalFagerbergPedersen03}.} Our faster $O(n^2)$
algorithm offers a 2-approximate solution when an exact value of
the parametric quartet distance is not required.  Additionally,
our algorithm gives the exact answer when $p = \frac{1}{2}$.

\paragraph{Related work.}
Several other measures for comparing trees have been proposed; we
mention a few. A popular class of distances are those based on
symmetric distance between sets of \emph{clusters} (that is, on
sets of species that descend from the same internal node in a
rooted tree) or of \emph{splits} (partitions of the set of species
induced by the removal of an edge in an unrooted tree); the latter
is the well-known Robinson-Foulds (RF) distance
\cite{RobinsonFoulds81}. It is not hard to show that two rooted
(unrooted) trees can share many triplet (quartet) topologies but
not share a single cluster (split). Cluster- and split-based
measures are also coarser than triplet and quartet distances.
%; for example, cluster distance can assume only $O(n)$
%distinct values versus the $O(n^3)$ possible ones for triplet
%distance.
%Furthermore, the set
%of quartets (or rooted triples) shared by two trees typically
%contains much more information than the set of common splits or
%clusters.

One can also measure the distance between two trees by counting
the number of \emph{branch-swapping} operations --- e.g.,
nearest-neighbor interchange or subtree pruning and regrafting
operations \cite{Felsenstein03} --- needed to convert one of the
trees into the other \cite{AllenSteel01}. However, the associated
measures can be hard to compute, and they fail to distinguish
between operations that affect many species and those that affect
only a few.  An alternative to distance measures are
\emph{similarity} methods such as maximum agreement subtree (MAST)
approach \cite{FindenGordon85}.  While there are efficient
algorithms for computing the MAST \cite{FarachThorup94}, the
measure is coarser than triplet-based distances.

There is an extensive literature on consensus methods for
phylogenetic trees. A non-exhaustive list of methods based on
splits or clusters includes strict consensus trees
\cite{McMorrisMeronkNeumann83}, majority-rule trees
\cite{BarthelemyMcMorris86}, and the Adams
consensus~\cite{Adams86}.  In \emph{local consensus} methods, the
goal is to find a consensus tree that satisfies a given set of
constraints on the topology of each triplet
\cite{KannanWarnowYooseph98}. For a thorough survey of these
methods, their properties and interrelationships, see
\cite{Bryant03}.

The fact that consensus methods tend to produce unresolved trees,
with an attendant loss of information, has been observed before.
An alternative approach is to provide multiple consensus trees,
instead of a single one. The idea, developed more fully in
\cite{StockhamWangWarnow02}, is to cluster the input trees using
some distance measure into groups, each of which is represented by
a single consensus tree, in such a way as to minimize some measure
of information loss. Our distance measures can be used within this
framework, where their fine-grained nature could conceivably offer
advantages over other techniques.

In addition to consensus methods, there are techniques that take
as input sets of quartet trees or triplet trees and try to find
large compatible subsets or subsets whose removal results in a
compatible set \cite{BerryJiangKearneyLiWareham99,SnirRao06}.
These problems are related to the \emph{supertree problem}, in
which a set of input trees that may not all share the same species
is given and the problem is to find a single tree that exhibits as
much as possible of the evolutionary relationships among the input
trees \cite{BinindaEmonds04}.  Thus, the consensus problem for
trees is a special case of the supertree problem.

The consensus problem on trees exhibits parallels with the
\emph{rank aggregation problem}, a problem with a rich history and
which has recently found applications to Internet search
%\cite{BartholdiToveyTrick89,Critchlow80,DiaconisGraham77,Kemeny59,YoungLevenglick78,DworkKumarNaorSivakumar01,FaginKumarMahdianSivakumarVee06}.
\cite{AilonCharikarNewman05,BartholdiToveyTrick89,Critchlow80,DiaconisGraham77,Kemeny59,DworkKumarNaorSivakumar01,FaginKumarMahdianSivakumarVee06}.
Here, we are given a collection of rankings (that is,
permutations) of $n$ objects, and the goal is to find a ranking of
minimum total distance to the input rankings. A distance between
rankings of particular interest is \emph{Kendall's tau}, defined
as the number of pairwise disagreements between the two rankings.
Like triplet and quartet distances, Kendall's tau is based on
elementary ordering relationships. Rank aggregation under
Kendall's tau was shown to be NP-complete even for four lists by
Dwork et al.~\cite{DworkKumarNaorSivakumar01}.

A permutation is the analog of a fully resolved tree, since every
pairwise relationship between elements is given.  The analog to a
partially-resolved tree is a \emph{partial ranking}, in which the
elements are grouped into an ordered list of \emph{buckets}, such
that elements in different buckets have known ordering
relationships, but elements within a bucket are not ranked
\cite{FaginKumarMahdianSivakumarVee06}.  Our definitions of
parametric distance and Hausdorff distance are inspired by Fagin
et al.'s \emph{Kendall tau with parameter $p$} and their Hausdorff
version of Kendall's tau,
respectively~\cite{FaginKumarMahdianSivakumarVee06}. We note,
however, that aggregating partial rankings seems computationally
easier than the consensus problem on trees.  For example, while
the Hausdorff version of Kendall's tau has a simple and
easily-computable expression
\cite{Critchlow80,FaginKumarMahdianSivakumarVee06}, it is unclear whether the
Hausdorff triplet or quartet distances are polynomially-computable
for trees.

\begin{comment}
Metrics lead to intriguing connections with the \emph{axiomatic}
approach to aggregation.  In this approach, one stipulates
desirable properties that a method should have and the goal is to
see if it a method with such properties exists
\cite{DayMcMorris03,Bryant03,SteelDressBocker00}.
\end{comment}

\paragraph{Organization of the paper.}
Section~\ref{sec:preliminaries} reviews basic notions in
phylogenetics and distances. Our distance measures and the
consensus problem are formally defined in
Section~\ref{sec:distances}.  The expected values of the distance measures are studied in Section~\ref{sec:expected}.  The basic properties of parametric
distance are proved in Section~\ref{sec:properties}.
Section~\ref{sec:relationships} studies the connection between
Hausdorff and parametric distances.  Section~\ref{sec:algorithms-triplets} gives efficient algorithms for computing parametric triplet distance.  A 2-approximation algorithm for parametric quartet distance is given in Section~\ref{sec:PQD}.

\section{Preliminaries}

\label{sec:preliminaries}

\paragraph{Phylogenies.}
By and large, we follow standard terminology (i.e., similar to
\cite{Bryant97} and \cite{SempleSteel03}).  We write $[N]$ to
denote the set $\{1,2, \dots, N\}$, where $N$ is a positive
integer.

Let $T$ be a rooted or unrooted tree.  We write $\V(T)$, $\E(T)$,
and $\Leaves(T)$ to denote, respectively, the node set, edge set,
and leaf set of $T$. A \emph{taxon} (plural \emph{taxa}) is some
basic unit of classification; e.g., a species. Let $\Species$ be a
set of taxa. A \emph{phylogenetic tree} or \emph{phylogeny} for
$\Species$ is a tree $T$ such that $\Leaves(T) = \Species$.
Furthermore, if $T$ is rooted, we require that every internal node
have at least two children; if $T$ is unrooted, every internal
node is required to have degree at least three.  We write
$\RPhyls(n)$ to denote the set of all rooted phylogenetic trees
over $\Species = [n]$ and $\Phyls(n)$ to denote the set of all
unrooted phylogenetic trees over $\Species = [n]$.

An internal node in a \emph{rooted} phylogeny is \emph{resolved}
if it has exactly two children; otherwise it is \emph{unresolved}.
Similarly, an internal node in an \emph{unrooted} phylogeny is
\emph{resolved} if it has degree three, and  \emph{unresolved}
otherwise.  Unresolved nodes in rooted and unrooted trees are also
referred to as \emph{polytomies} or \emph{multifurcations}.  A
phylogeny (rooted or unrooted) is \emph{fully resolved} if all its
internal nodes are resolved.  A \emph{fan} is a completely unresolved phylogeny; i.e., it contains a single internal node, to which all leaves are connected (if the phylogeny is rooted, this internal node is the root).

A \emph{contraction} of a phylogeny $T$ is obtained by deleting an
internal edge and identifying its endpoints.  A phylogeny $T_2$ is
a \emph{refinement} of phylogeny $T_1$, denoted $T_1 \preceq T_2$,
if and only if $T_1$ can be obtained from $T_2$ through $0$ or
more contractions.  Tree $T_2$ is a \emph{full refinement} of
$T_1$ if $T_1 \preceq T_2$ and $T_2$ is fully resolved.  We write
$\FullRef(T)$ to denote the set of all full refinements of $T$.

Let $X$ be a subset of $\Leaves(T)$ and let $T[X]$ denote the
minimal subtree of $T$ having $X$ as its leaf set.  The
\emph{restriction} of $T$ to $X$, denoted $T|X$, is the phylogeny
for $X$ defined as follows.  If $T$ is unrooted, then $T|X$ is the
tree obtained from $T[X]$ by suppressing all degree-two nodes.  If
$T$ is rooted, $T|X$ is obtained from $T[X]$ by suppressing all
degree-two nodes except for the root.

A \emph{triplet} is a three-element subset of $\Species$.  A
\emph{triplet tree} is a rooted phylogeny whose leaf set is a
triplet. The triplet tree with leaf set $\{a, b, c\}$ is denoted
by $a|bc$ if the path from $b$ to $c$ does not intersect the path
from $a$ to the root. A \emph{quartet} is a four-element subset of
$\Species$ and a \emph{quartet tree} is an unrooted phylogeny
whose leaf set is a quartet. The quartet tree with leaf set $\{a,
b, c, d\}$ is denoted by $ab|cd$ if the path from $a$ to $b$ does
not intersect the path from $c$ to $d$.  A triplet (quartet) $X$
is said to be \emph{resolved} in a phylogenetic tree $T$ over
$\Species$ if $T|X$ is fully resolved; otherwise, $X$ is
\emph{unresolved}.

Finally, we introduce notation for certain useful subtrees of a tree $T$.   Suppose $T$ is rooted and $v$ is a node
in $T$. Then, $T(v)$ denotes the subtree of $T$ rooted at $v$.
Suppose $T$ is unrooted and $\{u, v\}$ is an edge in $T$. Removal of edge $\{u,v\}$ splits the tree $T$ into two subtrees.
We denote the subtree that contains node $u$ by $T(u, v)$, and the subtree that contains $v$ by $T(v, u)$.

\paragraph{Distance measures, metrics, and near-metrics.}
A \emph{distance measure} on a set $D$ is a binary function $d$ on
$D$ satisfying the following three conditions: (i) $d(x,y) \geq 0$
for all $x, y \in D$; (ii) $d(x, y) = d(y, x)$ for all $x, y \in
D$; and (iii) $d(x, y) = 0$ if and only if $x = y$.  Function $d$
is a \emph{metric} if, in addition to being a distance measure, it
satisfies the triangle inequality; i.e., $d(x, z) \leq d(x, y) +
d(y, z)$ for all $x, y, z \in D$. Distance measure $d$ is a
\emph{near-metric} if there is a constant $c$, independent of the
size of $D$, such that $d$ satisfies the \emph{relaxed polygonal
inequality}: $d(x, z) \leq c(d(x, x_1) + d(x_1, x_2) + \dots +
d(x_{n-1}, z))$ for all $n > 1$ and $x, z, x_1, \dots , x_{n-1}
\in D$ \cite{FaginKumarMahdianSivakumarVee06}.  Two distance
measures $d$ and $d'$ with domain $D$ are \emph{equivalent} if
there are constants $c_1, c_2 > 0$ such that $c_1 d'(x,y) \leq
d(x,y) \leq c_2 d'(x,y)$ for every pair $x, y \in D$
\cite{FaginKumarMahdianSivakumarVee06}.

\section{Distance measures for phylogenies}

\label{sec:distances}

Here we define the distance measures for rooted and unrooted trees to be studied in the rest of the paper.  We use essentially the same notation for the rooted tree measures as for the unrooted tree measures.   We do so  because the concepts for each case are close analogs of those for the other, the key difference being the use  of triplets in one setting (rooted trees) and of quartets in the other (unrooted trees).  It will be easy to distinguish between the two settings by simply specifying the context in which the measures are being applied.   Our notation has the benefits of reducing repetitiveness and of allowing us to avoid excessive use of subscripts and superscripts.

Let $T_1$ and $T_2$ be any two rooted (respectively, unrooted)
phylogenies over taxon set $[n]$.  Define the
following five sets of triplets (quartets) over $[n]$.
\begin{description}
\item[$\SameRes(T_1,T_2)$:]  The set of all
triplets (quartets) $X$ such that $T_1 | X$ and $T_2|X$ are fully
resolved, and $T_1 | X = T_2|X$.
\item[$\DiffRes(T_1,T_2)$:]  The set of all
triplets (quartets) $X$ such that $T_1 | X$ and $T_2|X$ are fully
resolved, and $T_1 | X \neq T_2|X$.
\item[$\Res_1(T_1,T_2)$:]  The set of all
triplets (quartets) $X$ such that $T_1 | X$ is fully resolved, but
$T_2|X$ is not.
\item[$\Res_2(T_1,T_2)$:]  The set of all
triplets (quartets) $X$ such that $T_2 | X$ is fully resolved, but
$T_1|X$ is not.
\item[$\UnRes(T_1,T_2)$:]  The set of all
triplets (quartets) $X$ such that $T_1 | X$ and $T_2|X$ are
unresolved.
\end{description}

Let $p$ be a real number in the interval $[0,1]$.  The
\emph{parametric triplet (quartet) distance between $T_1$ and
$T_2$} is defined as\footnote{Note that the sets
$\SameRes(T_1,T_2)$ and $\UnRes(T_1,T_2)$ are not used in the
definition of $d^{(p)}$, but are needed for other purposes.}
\begin{equation}\label{eqn:param_tree_dist}
d^{(p)}(T_1,T_2) = |\DiffRes(T_1,T_2)| + p \left (
|\Res_1(T_1,T_2)| + |\Res_2(T_1,T_2)| \right ).
\end{equation}
When the domain of $d^{(p)}$ is restricted to fully resolved
trees, and thus $\Res_1(T_1,T_2) = \Res_2(T_1,T_2) =
\UnRes(T_1,T_2) = \emptyset$, we refer to it simply as the
\emph{triplet (quartet) distance}.

Parameter $p$ allows one to make a smooth transition from soft to
hard views of polytomy:  When $p = 0$, resolved triplets
(quartets) are treated as equal to unresolved ones, while when
$p=1$, they are treated as being completely different. Choosing
intermediate values of $p$ allows one to adjust for the amount of
evidence required to resolve a polytomy\footnote{We note that
parametric triplet/quartet distance is a \emph{profile-based
metric}, in the sense of \cite{FaginKumarMahdianSivakumarVee06}.
However, the use of the word ``profile'' in
\cite{FaginKumarMahdianSivakumarVee06} is quite different from our
use of the term.}.

An alternative distance measure (inspired by
References~\cite{FaginKumarMahdianSivakumarVee06,Critchlow80}), is the
\emph{Hausdorff distance}, defined as follows. Let $d$ be a metric
over fully resolved trees. Metric $d$ is extended to partially
resolved trees as follows.
\begin{equation}\label{eqn:hausdorff}
d_\text{Haus}(T_1,T_2) = \max \left\{ \max_{t_1 \in \FullRef(T_1)}
\min_{t_2 \in \FullRef(T_2)} d(t_1,t_2), \max_{t_2 \in
\FullRef(T_2)} \min_{t_1 \in \FullRef(T_1)} d(t_1,t_2) \right\}
\end{equation}
When $d$ is the triplet (quartet) distance, $d_\text{Haus}$ is
called the \emph{Hausdorff triplet (quartet) distance}.

Definition \eqref{eqn:hausdorff} requires some explanation. The
quantity $\min_{t_2 \in \FullRef(T_2)} d(t_1,t_2)$ is the distance
between $t_1$ and the set of full refinements of $T_2$. Hence,
\[\max_{t_1 \in \FullRef(T_1)} \min_{t_2 \in \FullRef(T_2)}
d(t_1,t_2)\] is the maximum distance between a full refinement of
$T_1$ and the set of full refinements of $T_2$. Similarly,
\[\max_{t_2 \in \FullRef(T_2)} \min_{t_1 \in \FullRef(T_1)}
d(t_1,t_2)\] is the maximum distance between a full refinement of
$T_2$ and the set of full refinements of $T_1$. Therefore, $T_1$
and $T_2$ are at Hausdorff distance $r$ of each other if every
full refinement of $T_1$ is within distance $r$ of a full
refinement of $T_2$ and vice-versa.

\paragraph{Aggregating phylogenies.}
Let $k$ be a positive integer and $\Species$ be a set of taxa.  A
\emph{profile of length $k$} (or simply a \emph{profile}, when $k$
is understood from the context) is a mapping $\Prof$ that assigns
each $i \in [k]$ a phylogenetic tree $\Prof(i)$ over $\Species$.
We refer to these trees as \emph{input trees}.  A \emph{consensus
rule} is a function that maps a profile $\Prof$ to some
phylogenetic tree $T$ over $\Species$ called a \emph{consensus
tree}.

Let $d$ be a distance measure whose domain is the set of
phylogenies over $\Species$. We extend $d$ to define a distance
measure from profiles to phylogenies as $d(T,\Prof) = \sum_{i=1}^k
d(T, \Prof(i)).$  A consensus rule is a \emph{median rule} for $d$
if for every profile $\Prof$ it returns a phylogeny $T^*$ of
minimum distance to $\Prof$; such a $T^*$ is called a
\emph{median}.  The problem of finding a median for a profile with
respect to a distance measure $d$ is referred to as the
\emph{median problem} (relative $d$), or as the \emph{aggregation
problem}.

\section{Expected parametric triplet and quartet distances}\label{sec:expected}

We now consider the expected value of parametric triplet and
quartet distances. Let $u(n)$ and $r(n)$ denote the probabilities
that a given quartet is, respectively, unresolved or resolved in
an unrooted phylogeny chosen uniformly at random from $\Phyls(n)$;
thus, $u(n) = 1 - r(n)$.  The following are the two main results of this section.

\begin{theorem}\label{thm:unrooted_distrib}
Let $T_1$ and $T_2$ be two unrooted phylogenies chosen uniformly
at random with replacement from $\Phyls(n)$.  Then,
\begin{equation}
E(d^{(p)}
(T_1,T_2) ) = {n \choose 4} \cdot \left ( \frac{2}{3} \cdot r(n)^2
+ 2 \cdot p \cdot r(n) \cdot u(n) \right ) .
\end{equation}
\end{theorem}

\begin{theorem}\label{thm:rooted_distrib}
Let $T_1$ and $T_2$ be two rooted phylogenies chosen uniformly at
random with replacement from $\RPhyls(n)$.  Then,
\begin{equation}
E(d^{(p)}
(T_1,T_2) ) = {n \choose 3} \cdot \left ( \frac{2}{3} \cdot
r(n+1)^2 + 2 \cdot p \cdot r(n+1) \cdot u(n+1) \right ) .
\end{equation}
\end{theorem}

It is known  \cite{Steel89,SteelPenny93} that
\begin{equation}
u(n) %= \sum_{i=0}^{n-3} {{n-3} \choose i} \frac{|P(i+2)|\cdot |\Phyls(n-i-1)|}{2|\Phyls(n)|}
\sim \sqrt{\frac{\pi (2 \ln 2 -1)}{4n}}
             \label{dist:eq2}.
\end{equation}
Together with
Theorems~\ref{thm:unrooted_distrib} and~\ref{thm:rooted_distrib},
this implies that $E(d^{(p)} (T_1,T_2) )$ is asymptotically
$\frac{2}{3} \cdot {n \choose 4}$ for unrooted trees and
$\frac{2}{3} \cdot {n \choose 3}$ for rooted trees.

The proof of Theorem~\ref{thm:unrooted_distrib} follows directly from the work of Day
\cite{Day86}; hence, it is omitted (however, we should note that the proof is similar to that of Lemma~\ref{dist:lemma1} below).  In the remainder of this section, we give a proof of Theorem~\ref{thm:rooted_distrib}.

We need some notation.  Let $u'(n)$ and $r'(n)$ denote the probabilities that a given
triplet is, respectively, unresolved or resolved in an rooted
phylogeny chosen at random from $\RPhyls(n)$.

\begin{lemma} \label{dist:lemma1}
Let $T_1$ and $T_2$ be two rooted phylogenies chosen uniformly at
random with replacement from $\RPhyls(n)$.  Then,
\begin{equation} \label{dist:equation1}
E(d^{(p)} (T_1,T_2) ) = {n \choose 3} \cdot \left ( \frac{2}{3}
\cdot r'(n)^2 + 2 \cdot p \cdot r'(n) \cdot u'(n) \right ) .
\end{equation}
\end{lemma}

\begin{proof}
By the definition of $d^{(p)}$ and the linearity of expectation,
it suffices to establish the equalities below.
\begin{equation}\label{eqn:DiffRes}
E(\DiffRes(T_1, T_2)) = {n \choose 3} \cdot \frac{2}{3} \cdot
r'(n)^2
\end{equation}
\begin{equation}\label{eqn:Res}
E(\Res_1(T_1, T_2)) = E(\Res_2(T_1, T_2)) = {n \choose 3} \cdot
r'(n) \cdot u'(n))
\end{equation}

To establish Equation~\eqref{eqn:DiffRes}, consider a triplet $X$.
The probability that $X$ is resolved in $T_1$ (or $T_2$) is
$r'(n)$. Thus, the probability that $X$ is resolved in both $T_1$
and $T_2$ is $r'(n)^2$. There are exactly three different ways in
which any given triplet can be resolved. Hence, if $\alpha$ is
resolved in both $T_1$ and $T_2$, the probability that it is
resolved differently in both trees is $\frac{2}{3}$. Thus, the
probability of a pre-given triplet being resolved in both $T_1$
and $T_2$, but with different types in each, is $\frac{2}{3}
r'(n)^2$. By the linearity of expectation and since
the total number of triplets from $\Leaves(T_1)$ (and $\Leaves(T_2)$) is ${n \choose
3}$, $E(\DiffRes(T_1, T_2)) = {n \choose 3} \cdot \frac{2}{3}
r'(n)^2$.

To establish Equation~\eqref{eqn:Res},  we only need to study
$E(\Res_1(T_1, T_2))$; the expression for $E(\Res_2(T_1, T_2))$
follows by symmetry. Consider a triplet $X$. The probability that
$X$ is unresolved in $T_1$ is $u'(n)$ and the probability that $X$
is resolved in $T_2$ is $r'(n)$. The expression for $E(\Res_1(T_1,
T_2))$ now follows by linearity of expectation.
\end{proof}

Let us define the function $\textsc{Add-Leaf}:  \RPhyls(n) \rightarrow \Phyls(n+1)$ as follows.
Given a rooted tree $T \in \RPhyls(n)$, $\textsc{Add-Leaf}(T)$ is the unrooted tree constructed from $T$ by (1) adding a leaf node labeled
$n+1$ to $T$ by adjoining it to the root node of $T$ and (2) unrooting
the resulting tree.
The next two lemmas are well known (for proofs, see~\cite{Steel89, Felsenstein03} and~\cite[p.~20]{SempleSteel03}, respectively).

\begin{lemma} \label{dist:lem1}
For all $n \geq 1$, $|\RPhyls(n)| = |\Phyls(n+1)|$.
\end{lemma}

\begin{lemma} \label{dist:lemma3}
Function \textsc{Add-Leaf} is a bijection from the set
$\RPhyls(n)$ to the set $\Phyls(n+1)$.
\end{lemma}

For any triplet $X$ over $[n]$, we
define two functions $g_X \colon \RPhyls(n) \rightarrow
\{0, 1\}$ and $f_X \colon \Phyls(n+1) \rightarrow \{0, 1\}$ as follows:
\begin{align}
\label{eq:g} g_X(T) & =
\begin{cases}
1  & \text{ if triplet } X \text{ is resolved in tree } T\\
0 & \text{ otherwise }
\end{cases} \\
\label{eq:f} f_X(T) & =
\begin{cases}
1  & \text{ if quartet } X \cup \{n+1\} \text{ is resolved in tree } T\\
0 & \text{ otherwise }
\end{cases}
\end{align}

We have the following result.
\begin{lemma} \label{dist:obs1}
Let $X$ be any triplet over $[n]$. Consider a tree $T \in \RPhyls(n)$, and let $T' = \textsc{Add-Leaf} (T)$.
Then, $f_X(T') = g_X(T)$.
\end{lemma}

\begin{proof}
Follows from the observation that triplet $X$ is resolved in $T$ if and only if
quartet $X \cup \{n+1\}$ is resolved in $T'$.
%\begin{proof}
%Consider a resolved triplet $\{a, b, c\}$, ($ 1 \leq a< b <c \leq
%n$), in $T$. Then in $T'$, $\{a, b, c, n+1\}$ must be a resolved
%quartet. Similarly, consider any resolved quartet $\{a, b, c,
%n+1\}$, ($1 \leq a < b < c \leq n$) in $T'$ that contains the leaf
%labeled $n+1$. Then in $T$, $\{a, b, c\}$ must be a resolved
%triplet.
%\end{proof}
\end{proof}

\begin{lemma}\label{lemma:ptob_triplet_quartet_bijection}
For all $n \geq 1$, $r'(n) = r(n+1)$ and $u'(n) = u(n+1)$.
\end{lemma}

\begin{proof}
Let $X$ be any triplet over $[n]$.
By definition, $r(n+1)$
is the probability of any given quartet being resolved in a random
unrooted tree in $\Phyls(n)$.  In particular, $r(n+1)$ is the probability that quartet
$X \cup \{n+1\}$ is resolved in a random
unrooted tree.  Now,
\begin{align*}
r(n+1) & =  \sum_{T \in \Phyls(n+1)} \frac{f_X(T)}{|\Phyls(n+1)|}  \\
& =  \sum_{T \in \Phyls(n+1)} \frac{f_X(T)}{|\RPhyls(n)|} \\
& =  \sum_{T' \in \RPhyls(n)} \frac{g_X(T')}{|\RPhyls(n)|} \\
& =  r'(n) ,
\end{align*}
where the first and last equalities follow from the definitions of $r(n+1)$ and $r(n)$, respectively, the second equality follows from Lemma~\ref{dist:lem1}, and the third follows from Lemma~\ref{dist:lemma3} and Lemma~\ref{dist:obs1}.

Since $u'(n) = 1 - r'(n)$ and $u(n+1) = 1 - r(n+1)$, it follows that $u'(n) = u(n+1)$.
\end{proof}

\begin{proof}[Proof of Theorem~\ref{thm:rooted_distrib}]
Simply substitute the expressions for $r'(n)$ and $u'(n)$ given in Lemma~\ref{lemma:ptob_triplet_quartet_bijection} into the expression for $E(d^{(p)} (T_1,T_2) )$ given in Lemma~\ref{dist:lemma1}.
\end{proof}

\section{Properties of parametric distance} \label{sec:properties}

In what follows, unless mentioned explicitly, whenever we refer to
parametric distance, we mean both its triplet and quartet
varieties. We begin with a useful observation.

\begin{proposition}\label{prop:dp_equiv}
For every $p, q$ such that $p, q \in (0,1]$, $d^{(p)}$ and
$d^{(q)}$ are equivalent.
\end{proposition}
\begin{proof}
Let  $T_1$ and $T_2$ be two rooted (unrooted) trees.  Let $M$ be the number of triplets (quartets) resolved differently in $T_1$ and let $N$ be the number of triplets (quartets) resolved only in one of $T_1$ and $T_2$. Then,  $d^{(p)}(T_1,T_2) = M+pN$, and $d^{(q)}(T_1,T_2) = M+qN$. Without loss of generality, let $p \ge q$. Now, if $c_1 = q/p$, then we have $c_1 d^{(q)}(T_1,T_2) = qM/p + q^2N/p \le M+pN = d^{(p)}(T_1,T_2)$. Similarly, if $c_2 = p/q$, then we have $c_2 d^{(q)}(T_1,T_2) = pM/q+pN \ge M+pN=d^{(p)}(T_1,T_2)$. Thus,  $c_1 d^{(q)}(T_1,T_2) \le d^{(p)}(T_1,T_2) \le c_2 d^{(q)}(T_1,T_2)$, and, consequently, $d^{(p)}$ and $d^{(q)}$ are equivalent.
\end{proof}

The next result precisely characterizes the ranges of $p$ for
which $d^{(p)}$ is a metric or near-metric:

\begin{theorem} \label{theorem:basic_dist_properties_1}\hspace{0cm}

\begin{enumerate}[(i)]
\item
For $p = 0$, $d^{(p)}$ is not a distance measure.
\item
For $p \in (0, 1/2)$, $d^{(p)}$ is a distance measure, but not a metric.
\item
For $p \in [1/2,1]$, $d^{(p)}$ is a metric.
\item
For $p \in (0, 1/2)$, $d^{(p)}$ is a near-metric.
\end{enumerate}
\end{theorem}

\begin{proof}
Our proof is analogous to the proof of the corresponding result for partial rankings given by Fagin et al.~\cite{FaginKumarMahdianSivakumarVee06}. For the sake of completeness, we prove this result formally.  For concreteness, we state our arguments in terms of rooted trees and triplets.  The extension to unrooted trees and quartets is direct.

For the proof of (i) and (ii), we use the same three triplet trees, $t_1 = ab|c$, $t_2 = abc$ (i.e., a completely unresolved tree), and $t_3 = ac|b$.
To prove (i), we note that $d^{(0)}(t_1,t_2) = 0$, even though $t_1 \neq t_2$.  Thus $d^{(0)}$ is not a distance measure.  Observe also that $d^{(0)}$ violates the triangle inequality, since $d^{(0)}(t_1,t_2)+d^{(0)}(t_2,t_3)= 2p = 0 < 1 = d^{(0)}(t_1,t_3)$.

To prove (ii), observe that $d^{(p)}(t_1,t_2) = d^{(p)}(t_2,t_3)= p$, and $d^{(p)}(t_1,t_3)=1$.  Thus, $d^{(p)}(t_1,t_3)=1 >2p = d^{(p)}(t_1,t_2)+d^{(p)}(t_2,t_3)$, violating the triangle inequality.  Thus, $d^{(p)}$ is not a metric in this case.  On the other hand, it is straightforward to verify that for any $p \in (0,1/2)$ --- as well, indeed, as for any $p \in [1/2,1]$ --- and any trees $T_1$ and $T_2$, we have $d^{(p)}(T_1,T_2) \ge 0$,  $d^{(p)}(T_1,T_2) = d^{(p)}(T_2,T_1)$, and $d^{(p)}(T_1,T_2)=0$ if and only if $T_1 = T_2$.  Thus, $d^{(p)}$ is a distance measure in this case.

We now prove (iii).  As mentioned in the proof of part (ii), $d^{(p)}$ is a distance measure for $p \in [1/2,1]$.  To complete the proof, we show that the triangle inequality holds; i.e., $d^{(p)}(T_1,T_3) \le d^{(p)}(T_1,T_2) + d^{(p)}(T_2,T_3)$ for any three trees $T_1, T_2, T_3$. Note that for any $i,j \in \{1,2,3\}$, we can express $d^{(p)}(T_i,T_j)$ as
\[d^{(p)}(T_i,T_j) = \sum_{\{a,b,c\} \subseteq [n]} d^{(p)}(T_i|\{a,b,c\},T_j|\{a,b,c\}).\]
That is, the distance between $T_i$ and $T_j$ can be expressed as the sum of parametric distances between all possible triplet trees induced by $T_i$ and $T_j$.  For any $\{a,b,c\} \subseteq [n]$, and each $i \in \{1,2,3\}$, let $t_i = T_i|\{a,b,c\}$.  To complete the proof of (iii), it suffices to show that $d^{(p)}(t_1,t_3) \le  d^{(p)}(t_1,t_2) + d^{(p)}(t_2,t_3)$.   If $t_1 = t_3$, then $d^{(p)}(t_1,t_3) = 0 \le d^{(p)}(t_1,t_2) + d^{(p)}(t_2,t_3)$, since distances are nonegative.  If $t_1 \neq t_3$, then $d^{(p)}(t_1,t_3) \le 1$, while $d^{(p)}(t_1,t_2) + d^{(p)}(t_2,t_3) \ge 2p$.   Thus, $d^{(p)}(t_1,t_3) \le  d^{(p)}(t_1,t_2) + d^{(p)}(t_2,t_3)$ if $p \in [1/2,1]$.

Finally, we prove (iv).  By Proposition~\ref{prop:dp_equiv}, for every $p \in (0,1/2)$, $d^{(p)}$ is equivalent to $d^{(1/2)}$, which, by part (iii), is a metric.  The claim now follows from a result by Fagin et al.~\cite{FaginKumarSivakumar03a} that implies that a distance measure is a near metric if and only if it is equivalent to a metric.
\end{proof}

Part (iii) of Theorem~\ref{theorem:basic_dist_properties_1} leads
directly to approximation algorithms: Let $\Prof$ be a profile, let $T^*$ be the median tree for $\Prof$, and let $T  = \Prof(\ell)$, where $\ell = \argmin_i d(\Prof(i),\Prof)$.  Then, by a standard approximation bound argument~(e.g., like those found in~\cite{Vazirani01}), we have that $d(T, \Prof) \le 2 d(T^*, \Prof)$.
Part (iv) indicates that the measure degrades nicely, since, along with the $2$-approximation algorithm for $p \in [1/2,1]$ implied by (iii), it leads to constant factor approximation
algorithms for $p \in (0,1/2)$
(an analogous observation for aggregation of partial rankings is made in~\cite{FaginKumarMahdianSivakumarVee06}).

The next result establishes a threshold for $p$ beyond which a
collection of fully resolved trees give enough evidence to produce
a fully resolved tree, despite the disagreements among them.

\begin{theorem} \label{theorem:basic_dist_properties_2}
Let $\Prof$ be a profile of length $k$, such that for all $i \in
[k]$, tree $\Prof(i)$ is fully resolved.  Then, if $p \geq 2/3$,
there exists median tree $T$ for $\Prof$ relative to $d^{(p)}$
such that $T$ is fully resolved.
\end{theorem}

It is interesting to compare
Theorem~\ref{theorem:basic_dist_properties_2} with analogous results
for partial rankings. Consider the variation of Kendall's tau for
partial rankings in which a pair of items that is ordered in one
ranking but in the same bucket in the other contributes $p$ to the
distance, where $p \in [0,1]$. This distance measure is a metric
when $p \geq 1/2$~\cite{FaginKumarMahdianSivakumarVee06}.
Furthermore, if $p \geq 1/2$ the median ranking relative to this
distance (that is, the one that minimizes the total distance to
the input rankings) is a full ranking if the input consists of
full rankings~\cite{BartholdiToveyTrick89}.  In contrast,
Proposition~\ref{theorem:basic_dist_properties_1} and
Theorem~\ref{theorem:basic_dist_properties_2} show that, in the range
$p \in [1/2,2/3)$, parametric triplet or quartet distance are
metrics, but the median tree is not guaranteed to be fully
resolved even if the input trees are.  The intuitive reason is
that for rankings, there are only two possible outcomes for a
comparison between two elements, but there are three ways in which
a triplet or quartet may be resolved.  This opens up a potentially useful  range of values for $p$ wherein parametric
triplet/quartet distance is a metric, but where one can adjust for
the degree of evidence (or confidence) needed to resolve a node.

Our proof of Theorem~\ref{theorem:basic_dist_properties_2} relies on two lemmas, which make use of the two procedures below.

\begin{description}
\item[\textsc{Pull-Out}$(T,u)$:]
The arguments are a rooted phylogenetic tree $T$ and a non-root node $u$ in
$T$, whose parent, denoted by $v$, has 3 or more children.  The procedure returns a new tree $T'$ obtained from $T$ as follows.  Split $v$ into two nodes $v'$ and $v''$ such
that the parent of $v'$ equals the parent of $v$, the children of
$v'$ are $u$ and $v''$, and the children of $v''$ are all the
children of $v$ except for $u$.

\item[\textsc{Pull-2-Out}$(T,u_1,u_2)$:]
The arguments are an unrooted phylogenetic tree $T$ and two nodes $u_1, u_2$ sharing the
same neighbor $v$ whose degree is at least 4 in $T$.  The procedure returns a new tree $T'$ obtained from $T$ as follows.  Split $v$ into two nodes $v'$ and $v''$ such that the neighbors of $v'$ are $v''$, $u_1$, and $u_2$, the neighbors of $v''$ are $v'$ and the neighbors of $v$ except for $u_1$ and $u_2$.
\end{description}

In what follows, we write $T_i$ to denote $\Prof(i)$, the $i$-th tree in profile $\Prof$, for $i
\in [k]$.  We need to introduce separate but analogous concepts for rooted and unrooted trees.

Suppose $T$ is a rooted phylogenetic tree and let $v$ be any node in $T$ with at least 3 children, denoted
$u_1, u_2, \dots , u_d$.  For $q \in [d]$, let $T^{(q)} = \textsc{Pull-Out}(T,u_q)$ and let $L_q$
denote the set of triplets $X$ such that $T|X$ is not fully resolved but $T^{(q)}|X$ is fully resolved.  Define the following two quantities.
\begin{eqnarray}
f_q & = & \sum_{X \in L_q} |\{i \in [k]: T_i|X \text{ agrees with } T^{(q)}|X\}| \\
a_q & = & \sum_{X \in L_q} |\{i \in [k]: T_i|X \text{ disagrees
with } T^{(q)}|X\}|.
\end{eqnarray}
Informally, $f_q$ and $a_q$ are the number of \emph{votes} cast by the trees in profile $\Prof$ for and against the way the triplets in $L_q$ are resolved in $T^{(q)}$.  Indeed, note that, by assumption, every tree in profile $\Prof$ is fully resolved. Thus, for each triplet $X = \{x,y,z\}$ and every $i \in [k]$, $T_i|X$ must agree with exactly one of $x|yz$, $y|xz$, or $z|xy$.  Thus, there are $k$ votes associated with each triplet
$X$, some for, some against.

Now suppose $T$ is an unrooted phylogenetic tree.   Let $v$ be any node in phylogeny $T$ and let
$u_1, u_2, \dots , u_d$ be the neighbors of $v$.  For $q,r \in [d]$,
let $T^{(qr)} = \textsc{Pull-2-Out}(T,u_q,u_r)$ and let $L_{qr}$ denote the set of quartets $X$ such that $T|X$ is not fully
resolved but $T^{(qr)}|X$ is fully resolved.   Define the
following two quantities.
\begin{eqnarray}
f_{qr} & = & \sum_{X \in L_{qr}} |\{i \in [k]: T_i|X \text{ agrees with } T^{(qr)}|X\}| \\
a_{qr} & = & \sum_{X \in L_{qr}} |\{i \in [k]: T_i|X \text{ disagrees
with } T^{(qr)}|X\}|.
\end{eqnarray}

We have the following result.

\begin{lemma}\label{lemma:Pull_Out_Aux}
For the rooted case, there exists an index $q \in [d]$ such that $f_q \geq a_q/2$.
For the unrooted case, there exists two indices $q,r\in [d]$ such that $f_{qr} \geq a_{qr}/2$.
\end{lemma}

\begin{proof}
For the rooted case, let $L = \bigcup_{q=1}^d L_q$. Thus, $L$ consists of those
triplets that are unresolved in $T$, but resolved in
$T^{(q)}$, for some $q \in [d]$. Equivalently, $L$ consists
of those triplets whose elements are leaves from three different
subtrees of $v$.

Let $X = \{x,y,z\}$ be a triplet in $L$.  Assume that $x \in
\Leaves(T(u_q))$, $y \in \Leaves(T(u_r))$, and $z \in
\Leaves(T(u_s))$, where $q,r,s$ must be distinct indices in $[d]$.
Then, $X$ is in $L_q$, $L_r$, and $L_s$.

Consider any $i \in [k]$.  By assumption, $T_i|X$ is a fully
resolved triplet tree.  Assume
without loss of generality that
$T_i|X = x|yz$. Then, $T^{(q)}|X$ agrees with $T_i|X$, so
$T_i|X$ contributes $+1$ to $f_q$.  On the other hand, both
$T^{(r)}|X$ and $T^{(s)}|X$ disagree with $T_i|X$, so
$T_i|X$ contributes $+1$ to $a_r$ and $+1$ to $a_s$. Furthermore,
for any $t \not\in \{q,r,s\}$, $T_i|X$ contributes nothing
to $f_t$ or $a_t$, since the triplet tree $T^{(t)}|X$ is not
fully resolved. Therefore, we have the following equalities.
\begin{align}~\label{eqn:aq}
\sum_{q = 1}^d a_q & = 2k \cdot |L| \\
\label{eqn:fq} \sum_{q = 1}^d f_q & = k \cdot |L|
\end{align}

Now suppose that for all $q \in [d]$, $f_q < a_q/2.$ This yields the
following contradiction:
\begin{equation*}
k \cdot |L| = \sum_{q=1}^d f_q < \frac{1}{2} \sum_{q=1}^d a_q
 = k \cdot |L|.
\end{equation*}
Here, the first equality follows from Equation~\eqref{eqn:aq} and
the last equality follows from Equation~\eqref{eqn:fq}.  Thus,
there must be some $q \in [d]$ such that $f_q \geq a_q/2.$

Similarly, for the unrooted case, let $L = \bigcup_{q,r \in [d], q\ne r} L_{qr}$.
Thus, $L$ consists of those quartets that are unresolved in $T$, but resolved in $T^{(qr)}$, for some $q,r \in [d]$, $q \ne r$. Equivalently, $L$ consists of those quartets whose elements are leaves from four different neighboring subtrees of $v$.

Let $X = \{w,x,y,z\}$ be a quartet in $L$.  Assume that $w \in
\Leaves(T(u_q,v))$, $x \in \Leaves(T(u_r,v))$, $y \in \Leaves(T(u_s,v))$, and $z \in \Leaves(T(u_t,v))$, where $q,r,s,t$ must be distinct indices in $[d]$.
Then, $X$ is in $L_q$, $L_r$, $L_s$, and $L_t$.

Consider any $i \in [k]$.  By assumption, $T_i|X$ is a fully
resolved quartet tree. Assume, without loss of generality, that $T_i|X = wx|yz$. Then, $T^{(qr)}|X$ and $T^{(st)}|X$ agree with $T_i|X$, so
$T_i|X$ contributes $+1$ to $f_{qr}$ and $f_{st}$, respectively. This double contribution is due to the symmetry of quartets. On the other hand,
$T^{(qs)}|X$, $T^{(qt)}|X$, $T^{(rs)}|X$, and $T^{(rt)}|X$ disagree with $T_i|X$, so
$T_i|X$ contributes $+1$ to $a_{qs}$, $a_{qt}$, $a_{rs}$, and $a_{rt}$, respectively. Furthermore,
if at least one of $t_1, t_2 \not\in \{q,r,s,t\}$, then $T_i|X$ contributes nothing
to $f_{t_1 t_2}$ or $a_{t_1 t_2}$, since the quartet tree $T^{(t_1 t_2)}|X$ is not fully resolved. Therefore, similar to the rooted case, we have the following equalities.
\begin{align}~\label{eqn:aq2}
\sum_{\substack{q,r \in [d] \\ q \ne r}} a_{qr} & = 4k \cdot |L| \\
\label{eqn:fq2} \sum_{\substack{q,r \in [d] \\ q \ne r}} f_{qr} & = 2k \cdot |L|
\end{align}

Now suppose that for all $q,r \in [d]$, $q \ne r$, $f_{qr} < a_{qr}/2.$ This yields the following contradiction:
\begin{equation*}
2k \cdot |L| = \sum_{\substack{q,r \in [d] \\ q\ne r}} f_{qr} < \frac{1}{2} \sum_{\substack{q,r \in [d] \\ q \ne r}} a_{qr}
 = 2k \cdot |L|.
\end{equation*}
Here, the first equality follows from Equation~\eqref{eqn:aq2} and
the last equality follows from Equation~\eqref{eqn:fq2}.  Thus,
there must be some $q,r \in [d]$, $q \ne r$, such that $f_{q,r} \geq a_{qr}/2.$
\end{proof}

%We have the following result for both rooted and unrooted cases.

\begin{lemma}\label{lem:Pull_Out_and_Pull_2_Out}
Let $\Prof$ be a profile for $[k]$ over $\Species$ consisting entirely of fully resolved rooted trees or fully resolved unrooted trees.  Let $T$ be a phylogeny for $\Species$; $T$ is rooted or unrooted according to whether $\Prof$ consists of rooted or unrooted trees.   Suppose $T$ contains an unresolved node $v$, and suppose $p \geq 2/3$.  Then,  the following holds.
\begin{enumerate}[(i)]
\item
If $T$ is rooted, $v$ has a child $u$ such that $d^{(p)}(\widehat{T},\Prof) \leq d^{(p)}(T,\Prof)$, where
$\widehat{T} = \textsc{Pull-Out}(T,u)$.
\item
If $T$ is unrooted, $v$ has two neighbors $u_q$ and $u_r$ such
that $d^{(p)}(\widehat{T},\Prof) \leq d^{(p)}(T,\Prof)$, where
$\widehat{T} = \textsc{Pull-2-Out}(T,u_q,u_r)$.
\end{enumerate}
\end{lemma}

\begin{proof}
We will show that in the rooted case, for all $q \in [d]$,
\begin{equation}\label{eqn:dist_reduction}
d^{(p)}(T^{(q)},\Prof) = d^{(p)}(T,\Prof) - p \cdot f_q +
(1-p) \cdot a_q.
\end{equation}
And, similarly, in the unrooted case,  for all $q,r \in [d]$,
\begin{equation}\label{eqn:dist_reduction_2}
d^{(p)}(T^{(qr)},\Prof) = d^{(p)}(T,\Prof) - p \cdot f_{qr} +
(1-p) \cdot a_{qr}.
\end{equation}

To verify this, consider any triplet or quartet $X \in L_q$.  For every $j$
such that $T^{(q)}|X$ or $T^{(qr)}|X$ is identical to $T_j|X$, the net
change in the distance from $\Prof$ is $-p$, since, for this $X$,
$T_j$ contributes $p$ to the distance to $T$, but contributes $0$
to the distance to $T^{(q)}$ or $T^{(qr)}$.  For every $j$ such that
$T^{(q)}|X$ or $T^{(qr)}|X$ is different from $T_j|X$, the net change in the
distance from $\Prof$ is $1-p$, since, for this $X$, $T_j$
contributes $p$ to the distance to $T$, but contributes $+1$ to
the distance to $T^{(q)}$ or $T^{(qr)}$.

Now, for the rooted case, choose an $q^* \in [d]$ such that $f_{q^*} \geq a_{q^*}/2$; for the unrooted case, choose two indices $q^*, r^*\in [d]$, $q^* \ne r^*$, such that $f_{q^* r^*} \geq a_{q^* r^*}/2$.  The existence of such a $q^*$ (or $q^*$ and $r^*$) is guaranteed by
Lemma~\ref{lemma:Pull_Out_Aux}.
Then, Equation~\eqref{eqn:dist_reduction} and $p \geq 2/3$ imply that
$d^{(p)}(T^{(q^*)},\Prof) \leq d^{(p)}(T,\Prof)$. Similarly,
Equation~\eqref{eqn:dist_reduction_2} and $p \geq 2/3$ imply that
$d^{(p)}(T^{(q^* r^*)},\Prof) \leq d^{(p)}(T,\Prof)$.
\end{proof}

\begin{proof}[Proof of
Theorem~\ref{theorem:basic_dist_properties_2}]If $\Prof$ consists
of only fully-resolved trees, then any phylogeny $T$ can be
transformed into a fully-resolved tree $T'$ such that
$d^{(p)}(T',\Prof) \leq d^{(p)}(T,\Prof)$ by doing the following.
First, let $T' = T$. Next, while $T'$ contains an unresolved node,
perform the following three steps:
\begin{enumerate}
\item
Pick any unresolved node
$v$ in $T'$.

\item
If $T$ is rooted, find a child $u$ of $v$ such that $d^{(p)}(\widehat{T},\Prof) \leq d^{(p)}(T,\Prof)$, where
$\widehat{T} = \textsc{Pull-Out}(T,u)$.
If $T$ is unrooted, find two neighbors $u_q$, $u_r$ of $v$ such that
$d^{(p)}(\widehat{T},\Prof) \leq d^{(p)}(T,\Prof)$, where
$\widehat{T} = \textsc{Pull-2-Out}(T, u_q, u_r)$.

\item
Replace $T'$ by $\widehat{T}$.
\end{enumerate}

Note that the existence of a node $u$ such as the one required in Step 2 is guaranteed by Lemma~\ref{lem:Pull_Out_and_Pull_2_Out}. Thus, for $p \geq 2/3$, there always exists a fully-resolved median tree
relative to $d^{(p)}$.
\end{proof}

The proof of Theorem~\ref{theorem:basic_dist_properties_2} implies that if $p >
2/3$ and the input trees are fully resolved, the median tree
relative to $d^{(p)}$ \emph{must} be fully resolved.  On the other
hand, it is easy to show that when $p \in [1/2,2/3)$, there are
profiles of fully resolved trees whose median tree is only
partially resolved.

\section{Relationships among the metrics}\label{sec:relationships}

We do not know whether the Hausdorff triplet or Hausdorff quartet distances
are computable in polynomial time.  Indeed, we suspect that,
unlike their counterparts for partial rankings, this may not be
possible.  On the positive side, we show here that, in a broad
range of cases, it is possible to obtain an approximation to the
Hausdorff distance by exploiting its connection with parametric
distance. As in the previous section, our results apply to both
triplet and quartet distances. Our first result, which is proved later in this section, is as follows.

\begin{lemma}\label{lem:Hausdorff_lower_bound}
For every two phylogenies $T_1$ and $T_2$ over the same set of taxa,
\[d_{\mathrm{Haus}}(T_1,T_2) \geq |\DiffRes(T_1,T_2)| +
\frac{2}{3}\cdot \max \{|\Res_1(T_1,T_2)|, |\Res_2(T_1,T_2)|\}.\]
\end{lemma}

An upper bound on $d_{\mathrm{Haus}}$ is obtained by assuming that
$T_1$ and $T_2$ are refined so that the triplets (quartets) in
$\Res_1(T_1,T_2)$, $\Res_2(T_1,T_2)$, and $\UnRes(T_1,T_2)$ are
resolved differently in each refinement.  This gives us the following result, which we state without proof.

\begin{lemma}\label{lem:Hausdorff_upper_bound}
For every two phylogenies $T_1$ and $T_2$ over  the same set of taxa,
\[d_{\mathrm{Haus}}(T_1,T_2) \leq |\DiffRes(T_1,T_2)| +
|\Res_1(T_1,T_2)| + |\Res_2(T_1,T_2)| + |\UnRes(T_1,T_2)|.\]
\end{lemma}

It is instructive to compare
Lemmas~\ref{lem:Hausdorff_lower_bound} and
\ref{lem:Hausdorff_upper_bound} with the situation for partial
rankings.  The Hausdorff version of Kendall's tau is obtained by
viewing each partial ranking as the set of all possible full
rankings that can be obtained by refining it (that is, ordering
elements within buckets). The distance is then the Hausdorff
distance between the two sets, where the distance between two
elements is the Kendall tau score. Critchlow~\cite{Critchlow80}
has given exact bounds on this distance measure, which allow it to
be computed efficiently and to establish an equivalence with the
parametric version of Kendall's tau defined in
Section~\ref{sec:properties} \cite{FaginKumarMahdianSivakumarVee06}.
To be precise, let $L_1$ and $L_2$ be two partial rankings.
Re-using notation, let $\DiffRes(L_1,L_2)$ be the set of all pairs
that are ordered differently in $L_1$ and $L_2$, $\Res_1(L_1,L_2)$
be the set of pairs that are ordered in $L_1$ but in the same
bucket in $L_2$, and $\Res_2(L_1,L_2)$ be the set of pairs that
are ordered in $L_2$ but in the same bucket in $L_1$.  Then, it
can be shown that $d_{\textrm{Haus}}(L_1,L_2) =
|\DiffRes(L_1,L_2)| + \max\{|\Res_1(L_1,L_2)|, |\Res_2(L_1,L_2)|
\}$ (see \cite{Critchlow80,FaginKumarMahdianSivakumarVee06}).

It seems unlikely that a similar simple expression can be obtained
for Hausdorff triplet or quartet distance.  There are at least two
reasons for this. Let $L_1$ and $L_2$ be partial rankings. Then,
it is possible to resolve $L_1$ so that it disagrees with $L_2$ in
any pair in $\Res_2(L_1,L_2)$. Similarly, there is a way to
resolve $L_2$ so that it disagrees with $L_1$ in any pair in
$\Res_1(L_1,L_2)$.  We have been unable to establish an analog of this property for trees; hence, the $2 \over 3$ factor in
Lemma~\ref{lem:Hausdorff_lower_bound}. The second reason is due to
the properties of the set $\UnRes(L_1,L_2)$.  It can be shown that
is one can refine rankings $L_1$ and $L_2$ in such a way that pairs of
elements that are unresolved in both rankings are resolved the
same way in the refinements. This seems impossible to do, in general, for
trees and leads to the presence of $|\UnRes(T_1,T_2)|$ in
Lemma~\ref{lem:Hausdorff_upper_bound}.

The above observations prevent us from establishing equivalence
between $d_{\mathrm{Haus}}$ and $d^{(p)}$, although they do not
disprove equivalence either. In any event, the next result shows that when the number of
triplets (quartets) that are unresolved in both trees is suitably
small, equivalence \emph{does} hold.

\begin{theorem}
Let $\beta$ be a positive real number.   Then, for every $p \in (0,1]$, Hausdorff distance
and parametric distance are equivalent when restricted
to pairs of trees $(T_1, T_2)$ such that
$|\UnRes(T_1,T_2)| \leq \beta (|\DiffRes(T_1,T_2)| +
|\Res_1(T_1,T_2)| + |\Res_2(T_1,T_2)|).$
\end{theorem}

\begin{proof}
By Proposition~\ref{prop:dp_equiv}, it suffices to show that $d_{\mathrm{Haus}}$ is equivalent to $d^{(2/3)}$.
Lemma~\ref{lem:Hausdorff_lower_bound} shows that $d^{(2/3)}(T_1,T_2) \le d_{\mathrm{Haus}}(T_1,T_2)$.   Thus, we only need to show that, under our assumption about $|\UnRes(T_1,T_2)|$, there is some $c$ such that $d_{\mathrm{Haus}}(T_1,T_2) \le c \cdot d^{(2/3)}(T_1,T_2)$.  The reader can verify that the result follows by choosing $c = 3+3\beta$ and invoking
Lemma~\ref{lem:Hausdorff_upper_bound}.
\end{proof}

The remainder of this section is devoted to the proof of Lemma~\ref{lem:Hausdorff_lower_bound}.  The argument
proceeds in two steps.  First, we show that $T_1$ can be
refined so that it disagrees with $T_2$ in at least two thirds of the triplets (quartets) in $\Res_2(T_1,T_2)$. Next, we show the existence of an analogous refinement of $T_2$. Note that the triplets (quartets) in $\DiffRes(T_1,T_2)$ are resolved differently in any refinements of $T_1$ and $T_2$.  This gives lower bounds for both arguments in the outer $\max$ of the definition of $d_{\mathrm{Haus}}(T_1,T_2)$ (Equation~\ref{eqn:hausdorff}) and yields the lemma.

Let $v$ be a node in $T_1$.  If $T_1$ is rooted, then, as in Section~\ref{sec:properties}, let $u_1, \dots, u_d$ denote the children of $v$ in $T_1$ and $T^{(q)}_1$ denote $\textsc{Pull-Out}(T,u_q)$. Define $\TripletMap_q(v)$ to be the set of all triplets $X \in \Res_2(T_1,T_2)$ such that (i) the lca of $X$ in $T_1$ is $v$ and
(ii) $T_1|X$ is unresolved but $T^{(q)}_1|X$ is fully resolved.  Let $\TripletMap(v) = \bigcup_{q = 1}^d \TripletMap_q(v)$.  Thus, $\TripletMap(v)$ is the set of triplets associated with $v$ that are resolved in $T_2$ but not in $T_1$.

If $T_1$ is unrooted, $u_1, \dots, u_d$ denote
the neighbors of $v$ in $T_1$ and $T^{(qr)}_1$ denotes $\textsc{Pull-2-Out}(T_1,u_{qr})$, where \textsc{Pull-2-Out} is the function defined in Section~\ref{sec:properties}. Define $\TripletMap_{qr}(v)$ to be the set of all quartets $X \in
\Res_2(T_1,T_2)$ such that (i) $T_1|X$ is a fan, (ii) the paths between any two distinct pairs of taxa in $X$ meet at $v$, and (iii) $T_1|X$ is unresolved but $T^{(qr)}_1|X$ is fully resolved.  Let $\TripletMap(v) = \bigcup_{q,r \in [d], q \ne r} \TripletMap_{qr}(v)$.  Thus, $\TripletMap(v)$ is the set of quartets associated with $v$ that are resolved in $T_2$ but not in $T_1$.

Define the following two sets for the rooted case.
\begin{eqnarray}
F_q & = & \{X \in \TripletMap_q(v) : T_2|X \text{ agrees with } T^{(q)}_1|X\} \\
A_q & = & \{X \in \TripletMap_q(v) : T_2|X \text{ disagrees with
} T^{(q)}_1|X\}.
\end{eqnarray}
Define the following two sets for the unrooted case.
\begin{eqnarray}
F_{qr} & = & \{X \in \TripletMap_{qr}(v) : T_2|X \text{ agrees with } T^{(qr)}_1|X\} \\
A_{qr} & = & \{X \in \TripletMap_{qr}(v) : T_2|X \text{ disagrees with
} T^{(qr)}_1|X\}.
\end{eqnarray}

The next result is, in a sense, a counterpart to Lemma~\ref{lemma:Pull_Out_Aux}.

\begin{lemma}\label{lem:bad_resolution}
For the rooted case, there exists an index $q \in [d]$ such that $|A_q| \geq 2 |F_q|$. For the unrooted case, there
exist two indices $q,r\in [d]$, $q \ne r$, such that $|A_{qr}| \geq 2 |F_{qr}|$.
\end{lemma}

\begin{proof}
We start with the rooted case.
Consider any triplet $X = \{x,y,z\}$ in $\TripletMap(v)$. Assume
that $x \in \Leaves(T_1(u_q))$, $y \in \Leaves(T_1(u_r))$, and $z \in
\Leaves(T_1(u_s))$, where $q,r,s$ must be distinct indices in $[d]$.
Thus, $X$ is in $\TripletMap_q(v)$, $\TripletMap_r(v)$, and
$\TripletMap_s(v)$.

By definition of $\TripletMap(v)$, $T_2|X$ is a fully resolved
triplet tree. Assume that $T_2|X = x|yz$. Then, $T^{(q)}_1|X$
agrees with $T_2|X$, so $X$ contributes exactly one element to
$F_q$. On the other hand, both $T^{(r)}_1|X$ and
$T^{(s)}_1|X$ disagree with $T_2|X$, so $X$ contributes
exactly one element to $A_r$ and one element to $A_s$.
Furthermore, for any $t \not\in \{q,r,s\}$, $X$
contributes nothing to $F_t$ or $A_t$,  since the triplet tree
$T^{(t)}_1|X$ is not fully resolved. Therefore, we have that
\begin{equation}~\label{eqn:AFq}
\sum_{q = 1}^d |A_q| = 2 \cdot |\TripletMap(v)| \qquad
\text{and} \qquad \sum_{q = 1}^d |F_q| = |\TripletMap(v)|.
\end{equation}

Assume that for all $q \in [d]$, $|F_q| > |A_q|/2.$ This and
\eqref{eqn:AFq} imply that
\begin{equation*}
|\TripletMap(v)| = \sum_{q=1}^d |F_q| > \frac{1}{2} \sum_{q=1}^d
|A_q| = |\TripletMap(v)|,
\end{equation*}
a contradiction.

We now consider the unrooted case.
Consider any quartet $X = \{w,x,y,z\}$ in $\TripletMap(v)$. Assume
that $w \in \Leaves(T_1(u_q,v))$, $x \in \Leaves(T_1(u_r,v))$, $y \in \Leaves(T_1(u_s,v))$, and $z \in \Leaves(T_1(u_t,v))$, where $q,r,s,t$ must be distinct indices in $[d]$.
Thus, $X$ is in $\TripletMap_{qr}(v)$, $\TripletMap_{qs}(v)$, $\TripletMap_{qt}(v)$, $\TripletMap_{rs}(v)$, $\TripletMap_{rt}(v)$ and $\TripletMap_{st}(v)$.

By definition of $\TripletMap(v)$, $T_2|X$ is a fully resolved
quartet tree. Assume that $T_2|X = wx|yz$. Then, $T^{(qr)}_1|X$ and $T^{(st)}_1|X$ agree with $T_2|X$, so $X$ contributes exactly one element to
$F_{qr}$ and $F_{st}$. On the other hand, $T^{(qs)}_1|X$, $T^{(qt)}_1|X$, $T^{(rs)}_1|X$ and
$T^{(rt)}_1|X$ disagree with $T_2|X$, so $X$ contributes
exactly one element to $A_{qs}$, $A_{qt}$, $A_{rs}$ and $A_{rt}$, respectively.
Furthermore, for any $j_1$ and $j_2 \not\in \{q,r,s,t\}$, $X$
contributes nothing to $F_{j_1 j_2}$ or $A_{j_1 j_2}$,  since the quartet tree
$T^{(j_1 j_2)}_1|X$ is not fully resolved. Therefore, we have that
\begin{equation}~\label{eqn:AFq1q2}
\sum_{\substack{q,r \in [d] \\ q \ne r}} |A_{qr}| = 4 \cdot |\TripletMap(v)| \qquad
\text{and} \qquad \sum_{\substack{q,r \in [d] \\ q \ne r}} |F_{qr}| = 2 \cdot |\TripletMap(v)|.
\end{equation}

Assume that for all $q,r \in [d]$, $|F_{qr}| > |A_{qr}|/2.$ This and
\eqref{eqn:AFq1q2} imply that
\begin{equation*}
2 \cdot |\TripletMap(v)| = \sum_{\substack{q,r \in [d] \\ q \ne r}} |F_{qr}| > \frac{1}{2} \sum_{\substack{q,r \in [d] \\ q \ne r}}
|A_{qr}| = 2 \cdot |\TripletMap(v)|,
\end{equation*}
a contradiction.
\end{proof}

\begin{proof}[Proof of Lemma~\ref{lem:Hausdorff_lower_bound}]
Define the following functions.  For any two phylogenies
$T_1,T_2$ over $\Species$, let
\begin{eqnarray}
d_{H1} (T_1,T_2) & = & \max_{t_1 \in \FullRef(T_1)} \min_{t_2 \in
\FullRef(T_2)} d(t_1,t_2), \\
d_{H2} (T_1,T_2) & = & \max_{t_2 \in \FullRef(T_2)} \min_{t_1 \in
\FullRef(T_1)} d(t_1,t_2) .
\end{eqnarray}

We show that
\begin{eqnarray}\label{eqn:dH1}
d_{H1} (T_1,T_2) & \geq &  |\DiffRes(T_1,T_2)| + \frac{2}{3}\cdot
|\Res_2(T_1,T_2)| \\%
\label{eqn:dH2} d_{H2} (T_1,T_2) &  \geq & |\DiffRes(T_1,T_2)| +
\frac{2}{3}\cdot|\Res_1(T_1,T_2)| .
\end{eqnarray}
Since $d_\text{Haus} (T_1,T_2) = \max \{ d_{H1}(T_1,T_2),
d_{H2}(T_1,T_2) \}$, this proves Lemma~\ref{lem:Hausdorff_lower_bound}.

By symmetry, it suffices to prove Inequality~\eqref{eqn:dH1}. Our
argument relies on two observations. First, note that if $T_1'$ is
a refinement of $T_1$ (but possibly not a full refinement), then,
$d_{H1} (T_1,T_2) \geq d_{H1} (T_1',T_2)$. This holds because
$\FullRef(T_1') \subseteq \FullRef(T_1)$. Second, for any two
phylogenies $T_1$ and $T_2$, $d_{H1} (T_1,T_2) \geq
|\DiffRes(T_1,T_2)|$.  This holds because for any $t_1 \in
\FullRef(T_1)$, $t_2 \in \FullRef(T_2)$, we have that
$\DiffRes(T_1,T_2) \subseteq \DiffRes(t_1,t_2)$, and (by
definition) $d(t_1,t_2) = |\DiffRes(t_1,t_2)|$.

By the preceding observations, if we prove that it is possible to construct a refinement $T_1'$ of $T_1$ such that
$|\DiffRes(T_1',T_2)| \geq |\DiffRes(T_1,T_2)| + \frac{2}{3} |\Res_2(T_1,T_2)|$, then  Inequality~\eqref{eqn:dH1} follows.  The idea is to find a refinement $T_1'$ of $T_1$ such that for at least two-thirds of the triplets or quartets $X \in \Res_2(T_1,T_2)$, we have that $T_1'|X \neq T_2|X$.  To obtain the desired refinement of $T_1$, we initially set $T_1' = T_1$ and then perform the following steps
while they apply:
\begin{enumerate}
\item
Pick an unresolved node $v$ in $T_1'$ such
that $\TripletMap'(v) \neq \emptyset$, where $\TripletMap'(v)$ is the set of triplets (quartets) associated with $v$ that are resolved in $T_2$ but not in $T_1'$. In the rooted case, let $u_1,
\dots, u_d$ be the children of $v$; in the unrooted case, let $u_1,
\dots, u_d$ be the neighbors of $v$.
\item
For rooted trees, find a $q \in [d]$ such that $|A_q| \geq 2 |F_q|$ (such a $q$ exists by Lemma~\ref{lem:bad_resolution}).  For unrooted trees, find $q,r \in [d]$ such that $|A_{qr}| \geq 2
|F_{qr}|$ (such $q,r$ exist by Lemma~\ref{lem:bad_resolution}).
\item
In the rooted case, set $T_1' = \textsc{Pull-Out}(T_1',u_q)$; in the unrooted case, set $T_1' = \textsc{Pull-2-Out}(T_1',u_q,u_r)$.
\end{enumerate}

When this algorithm terminates, $\TripletMap'(v) = \emptyset$
for every $v \in \V(T_1')$.  Thus, $\Res_2(T_1',T_2) = \emptyset$.
Furthermore, the choice of $q$ (or $q_1$ and $q_2$) in step (2) guarantees that
$|\DiffRes(T_1',T_2)| \geq |\DiffRes(T_1,T_2)| + {2 \over 3} \cdot
|\Res_2(T_1,T_2)|.$
\end{proof}

\section{Computing parametric triplet distance} \label{sec:algorithms-triplets}

In this section we show that the parametric triplet distance
(PTD), $d^{(p)}$, between two phylogenetic trees $T_1$ and $T_2$
over the same set of $n$ taxa can be computed in $O(n^2)$ time.

Before we outline our PTD algorithm, we need some notation. Let $T$ be a rooted phylogenetic tree.  Then, $R(T)$
denotes the set of all triplets that are resolved in $T$ and
$U(T)$ denotes the set of all triplets that are unresolved in $T$.

The next proposition is easily proved.

\begin{proposition} \label{prop:triplet_obs}
For any two phylogenies $T_1$, $T_2$ over the same set of taxa,
\begin{enumerate}[(i)]
\item
$|\Res_1(T_1, T_2)| + |\UnRes(T_1, T_2)| = |U(T_2)|$
\item
$|\Res_2(T_1, T_2)| + |\UnRes(T_1, T_2)| = |U(T_1)|$,
\item
$|\SameRes(T_1, T_2)| +  |\DiffRes(T_1, T_2)| + |\Res_1(T_1, T_2)|
= |R(T_1)|$. \end{enumerate}
\end{proposition}

By Prop.~\ref{prop:triplet_obs} and
Eqn.~\eqref{eqn:param_tree_dist}, the parametric distance between
$T_1$ and $T_2$ can be expressed as
\begin{equation}\label{eqn:param_dist_2}
d^{(p)}(T_1, T_2) = |R(T_1)| - |\SameRes(T_1, T_2)| + p \cdot
(|U(T_1)| - |U(T_2)|) + (2p - 1) \cdot |\Res_1(T_1, T_2)|.
\end{equation}

Our PTD algorithm proceeds as follows. After an initial $O(n^2)$ preprocessing step (Section~\ref{trip:precomp}), the algorithm computes $|R(T_1)|$, $|U(T_1)|$ and
$|U(T_2)|$ using a $O(n)$-time procedure (Section~\ref{trip:RU}). Next, it computes $|\SameRes(T_1, T_2)|$ and
$|\Res_1(T_1, T_2)|$. As described in Sections~\ref{trip:sect:TD}
and~\ref{sec:R1_alg}, this takes $O(n^2)$ time. Then, it uses
these values to compute $d^{(p)}(T_1, T_2)$, in $O(1)$ time, via
Equation~\eqref{eqn:param_dist_2}. To summarize, we have the following result.

\begin{theorem} \label{triplet:theorem2}
The parametric triplet distance $d^{(p)}(T_1, T_2)$ for two rooted
phylogenetic trees $T_1$ and $T_2$ over the same set of $n$ taxa
can be computed in $O(n^2)$ time.
\end{theorem}

In the rest of this section we use the following notation.  We
write $\rt(T)$ to denote the root node of a tree $T$.  Let $v$ be a node
in $T$. Then, $\pa(v)$ denotes the parent of $v$ in $T$ and
$\ch(v)$ is the set of children of $v$.  We write $\overline{T(v)}$ to
denote the tree obtained by deleting $T(v)$ from $T$, as well as
the edge from $v$ to its parent, if such an edge exists.

%-------------- Start of Part 4 modification ---------------
\subsection{The preprocessing step} \label{trip:precomp}

The purpose of the preprocessing step is to calculate and store the following four quantities for every pair $(u, v)$, where $u \in \V(T_1)$ and $v \in \V(T_2)$: $|\Leaves(T_1(u)) \cap
\Leaves(T_2(v))|$, $|\Leaves(T_1(u)) \cap
\Leaves(\overline{T_2(v)})|$, $|\Leaves(\overline{T_1(u)}) \cap
\Leaves(T_2(v))|$, and $|\Leaves(\overline{T_1(u)}) \cap
\Leaves(\overline{T_2(v)})|$. These values are stored in a table so that any value can be accessed in $O(1)$ time by subsequent steps of the PTD algorithm.

\begin{lemma} \label{triplet:lemma5}
The values $|\Leaves(T_1(u)) \cap
\Leaves(T_2(v))|$, $|\Leaves(T_1(u)) \cap
\Leaves(\overline{T_2(v)})|$, $|\Leaves(\overline{T_1(u)}) \cap
\Leaves(T_2(v))|$, and $|\Leaves(\overline{T_1(u)}) \cap
\Leaves(\overline{T_2(v)})|$ can be collectively computed for every pair of nodes $(u, v)$, where $u \in \V(T_1)$ and  $v \in \V(T_2)$, in $O(n^2)$ time.
\end{lemma}
\begin{proof}
We first observe that for each $u \in \V(T_1)$, the value $|\Leaves(T_1(u))|$ can be computed in $O(n)$ time by a simple post order traversal of $T_1$. The same holds for tree $T_2$.

Consider the value $|\Leaves(T_1(u)) \cap \Leaves(T_2(v))|$.  We consider three cases.
\begin{enumerate}
\item  If $u$ and $v$ are both leaf nodes then computing $|\Leaves(T_1(u)) \cap \Leaves(T_2(v))|$ is trivial.

\item  If $u$ is a leaf node, but $v$ is not a leaf node, then $$|\Leaves(T_1(u)) \cap \Leaves(T_2(v))| = \sum_{x \in \ch(v)} |\Leaves(T_1(u)) \cap \Leaves(T_2(x))|.$$

\item If $u$ is not a leaf node, then $$|\Leaves(T_1(u)) \cap \Leaves(T_2(v))| = \sum_{x \in \ch(u)} |\Leaves(T_1(x)) \cap \Leaves(T_2(v))|.$$
\end{enumerate}

We compute the value $|\Leaves(T_1(u)) \cap \Leaves(T_2(v))|$, for every pair $(u, v)$, using an interleaved post order traversal of $T_1$ and $T_2$. This traversal works as follows: For each node $u$ in a post order traversal of $T_1$, we consider each node $v$ in a post order traversal of $T_2$. This ensures that when the intersection sizes for a pair of nodes is computed, the set intersection sizes for all pairs of their children have already been computed.
The total time complexity for computing the required values in this way can be bounded as follows. %The values $|\Leaves(T_1(u))|$  and $|\Leaves(T_2(v))|$ at each node $u \in T_1$ and $v \in T_2$ can be computed in $O(n)$ time.
For a pair of nodes $u$ and $v$ from $T_1$ and $T_2$ respectively, the value $|\Leaves(T_1(u)) \cap
\Leaves(T_2(v))|$ can be computed in $O(|\ch(u)| + |\ch(v)|)$ time and all the remaining three set intersection values in $O(1)$ time. Summing this over all possible pairs of edges, we get a total time of $O(\sum_{u \in \V(T_1)} \sum_{v \in \V(T_2)} |\ch(u)| + |\ch(v)|)$, which is $O(n^2)$.

Once the value $|\Leaves(T_1(u)) \cap
\Leaves(T_2(v))|$ has been computed for every pair $(u, v)$, the remaining quantities we seek can be computed using the following relations.
\begin{align*}
|\Leaves(T_1(u)) \cap \Leaves(\overline{T_2(v)})| & =  |\Leaves(T_1(u))| - |\Leaves(T_1(u)) \cap \Leaves(T_2(v))|,\\
|\Leaves(\overline{T_1(u)}) \cap \Leaves(T_2(v))| & =  |\Leaves(T_2(v))| - |\Leaves(T_1(u)) \cap \Leaves(T_2(v))|, \qquad \text{ and} \\
|\Leaves(\overline{T_1(u)}) \cap \Leaves(\overline{T_2(v)})| & = n - (|\Leaves(T_1(u))| + |\Leaves(T_2(v))| - |\Leaves(T_1(u)) \cap \Leaves(T_2(v))|).
\end{align*}
Thus, each of these values can be computed in $O(1)$ time, for a total of $O(n^2)$.
\end{proof}

We store these $O(n^2)$ values in an array indexed by $u$ and $v$, for each $u \in \V(T_1)$ and $v \in \V(T_2)$. This enables constant time insertion and look-up of any stored value, when the two relevant nodes are given.

%Remark: It is possible to store the $O(n^2)$ values computed in an array so that each value is stored at a fixed index depending on the node pair that the value is associated with. This enables constant time insertion and look-up of any stored value, when the two relevant nodes are known.

\subsection{Computing $|R(T_1)|$, $|U(T_1)|$, and $|U(T_2)|$} \label{trip:RU}

%The values $|R(T_1)|$, $|U(T_1)|$ and $|U(T_2)|$ are required to compute $d^{(p)}(T_1, T_2)$ according to Equation~\eqref{eqn:param_dist_2}.
Here we prove the following result.

\begin{lemma} \label{triplet:lemma1}
Given a rooted phylogenetic tree $T$ over $n$ leaves, the values $|R(T)|$
and $|U(T)|$ can be computed in $O(n)$ time.
\end{lemma}

Thus, $|R(T_1)|$, $|U(T_1)|$ and $|U(T_2)|$ can all be computed in $O(n)$ time.

To prove Lemma~\ref{triplet:lemma1}, we need some terminology and an auxiliary result. Let $e = (v, \pa(v))$ be any internal edge in $T$.
Consider any two leaves $x, y$ from $\Leaves(T(v))$, and any leaf
$z$ from $\Leaves(\overline{T(v)})$. Then, the triplet $\{x, y,
z\}$ must appear resolved as $xy|z$ in $T$; we say that the triplet tree $xy|z$ is \emph{induced} by the edge $(v, \pa(v))$. Note that
the same resolved triplet tree may be induced by multiple edges in $T$.
We say that the triplet tree $xy|z$ is \emph{strictly induced} by the edge $\{v, \pa(v)\}$ if $xy|z$ is induced by $(v, \pa(v))$ and, additionally, $x \in \Leaves(T(v_1))$ and $y \in \Leaves(T(v_2))$ for some $v_1, v_2 \in \ch(v)$ such that $v_1 \neq v_2$. %Note that the triplet tree $xy|z$ is strictly induced precisely by the edge $(v, \pa(v))$, where $v$ denotes the lca of $x$ and $y$

\begin{lemma} \label{lem:strictlyinduced}
Given a tree $T$ and a triplet $X$, if $T|X$ is fully resolved then $T|X$ is strictly induced by exactly one edge in $T$.
\end{lemma}
\begin{proof}
Let $X  = \{a, b, c\}$. Without loss of generality, assume that $T|X = ab|c$. %Then, by definition, the path from $a$ to $b$ in $T$ does not intersect the path from $a$ to the root.
If $v$ denotes the lca of $a$ and $b$ in $T$, the edge $\{v, \pa(v)\}$ must induce $ab|c$. Moreover, $v$ must be the only node in $T$ for which there exist nodes $v_1, v_2 \in \ch(v)$ such that $a \in \Leaves(T(v_1))$ and $b \in \Leaves(T(v_2))$. Thus, there is exactly one edge in $T$ that strictly induces $T|X$.
\end{proof}

\begin{proof}[Proof of Lemma~\ref{triplet:lemma1}]
Since $|R(T)| + |U(T)| = {n \choose 3}$, given $|R(T)|$, the value $|U(T)|$ can be computed in $O(1)$ additional time.  Thus, we only need to show that the value of $|R(T)|$ can be computed in $O(n)$ time.

The first step is to traverse the tree $T$ in
post order to compute the values $\alpha_v = |\Leaves(T(v))|$ and
$\beta_v = n - \alpha_v$  at each node $v \in \V(T)$.  This takes $O(n)$ time.

For any $v \in \V(T) \setminus \{\rt(T)\}$, let $\phi(v)$ denote the number of triplets that are strictly induced by the edge $\{v,\pa(v)\}$ in tree $T$. Observe that any triplet that is strictly induced by an edge in $T$ must be fully resolved in $T$. Thus, Lemma~\ref{lem:strictlyinduced} implies that the sum of $\phi(v)$ over all internal nodes $v \in \V(T) \setminus \{\rt(T)\}$ yields the value $|R(T)|$.  We now show how to compute the value of $\phi(v)$.

Let $X = \{a, b, c\}$ be a triplet that is counted in $\phi(v)$. And, without loss of generality, let $T_1|X = ab|c$. It can be verified that $X$ must satisfy the following two conditions: (i) $a, b \in \Leaves(T(v))$ and $c \in \Leaves(\overline{T(v)})$, and (ii) there does not exist any $x \in \ch(v)$ such that $a, b \in \Leaves(T(x))$.
The number of triplets that satisfy condition (i) is ${\alpha_v \choose 2} \cdot \beta_v$, and the number of triplets that satisfy condition (i), but not condition (ii) is exactly $\sum_{x \in \ch(v)} {\alpha_x \choose 2} \cdot \beta_v$. Thus, $\phi(v) =
\gamma_v - \sum_{x \in \ch(v)} {\alpha_x \choose 2} \cdot
\beta_v$.

Computing $\phi(v)$ requires $O(|\ch(v)|)$ time; hence, the time complexity for computing $|R(T)|$ is $O(\sum_{v \in \V(T)}|\ch(v)|)$, which is $O(n)$.
\end{proof}

\subsection{Computing $|\SameRes(T_1,T_2)|$} \label{trip:sect:TD}
%\subsection{Counting shared triplets} \label{trip:sect:TD}

We now describe an $O(n^2)$ time algorithm to compute the size of the set $\SameRes(T_1,T_2)$ of shared triplets; that is, triplets that are fully and identically resolved in $T_1$ and $T_2$.

For any $u \in \V(T_1) \setminus (\rt(T_1) \cup \Leaves(T_1))$ and $v \in \V(T_2) \setminus (\rt(T_2)\cup \Leaves(T_2))$, let $s(u, v)$ denote the number of identical triplet trees strictly induced by edge $\{u, \pa(u)\}$ in $T_1$ and edge $\{v, \pa(v)\}$ in $T_2$.  We have the following result.

%We show that each value $s(u, v)$ can be computed in  in $O(|\ch(u)| \cdot |\ch(v)|)$ time by using the values
%computed in the preprocessing step.
%A formal description of this algorithm appears in Algorithm~\ref{SharedTriplets}.

%We argue that the sum of all the $s(u, v)$ values is precisely $|\SameRes(T_1, T_2)|$, and that all the $s(u, v)$'s can be computed in $O(n^2)$ time.

\begin{lemma} \label{triplet:lemmaMainShared}
Given $T_1$ and $T_2$, we have,
\begin{equation} \label{Lemma:eqnMain}
%|\SameRes(T_1, T_2)| =  \sum_{u \in \V(T_1) \setminus (\rt(T_1) \cup \Leaves(T_1))} \sum_{v \in \V(T_2) \setminus (\rt(T_2) \cup \Leaves(T_2))} s(u,v).
|\SameRes(T_1, T_2)| =  \sum_{\substack{u \in \V(T_1) \setminus (\rt(T_1) \cup \Leaves(T_1)), \\ v \in \V(T_2) \setminus (\rt(T_2) \cup \Leaves(T_2))}} s(u,v).
\end{equation}
\end{lemma}
\begin{proof}
Consider any triplet $X \in \SameRes(T_1, T_2)$. %Without any loss of generality we may assume that $X = \{a, b, c\}$ and $T_1|X = T_2|X = ab|c$.
 Since $T_1|X$ is fully resolved and $T_1|X = T_2|X$ then, by Lemma~\ref{lem:strictlyinduced}, there exists exactly one node $u \in \V(T_1)\setminus \rt(T_1)$ and one node $v \in \V(T_2)\setminus \rt(T_2)$ such that the edge $\{u, \pa(u)\}$ strictly induces $T_1|X$ in $T_1$, and edge $\{v, \pa(v)\}$ strictly induces $T_2|X$ in $T_2$. Additionally, neither $u$ nor $v$ can be leaf nodes in $T_1$ and $T_2$ respectively.
Thus, $X$ would be counted exactly once in the right-hand side of Equation~\eqref{Lemma:eqnMain} in the value $s(u, v)$. Moreover, by the definition of $s(u, v)$, any triplet tree that is counted on the right-hand side of Equation~\eqref{Lemma:eqnMain} algorithm must belong to the set $\SameRes(T_1, T_2)$. %be resolved, and resolved identically, in $T_1$ and $T_2$.
The Lemma follows.
\end{proof}

The following lemma shows how to compute the value of $s(u, v)$ using the values computed in the preprocessing step.

%We now show that the value $s(u, v)$ can be computed in $O(\ch(u) \cdot \ch(v))$ time.

\begin{lemma} \label{lem:Shared_count}
Given any $u \in \V(T_1) \setminus (\rt(T_1) \cup \Leaves(T_1))$ and $v \in \V(T_2) \setminus (\rt(T_2)\cup \Leaves(T_2))$, $s(u,v)$ can be computed in $O(|\ch(u)|\cdot |\ch(v)|)$ time.
\end{lemma}
\begin{proof}
We will show that $s(u, v) = n_1(u, v) - n_2(u, v) - n_3(u, v) + n_4(u, v)$, where
\begin{align*}
n_1(u, v) & = {|\Leaves(T_1(u)) \cap \Leaves(T_2(v))| \choose 2} \cdot
|\Leaves(\overline{T_1(u)}) \cap \Leaves(\overline{T_2(v)})|,\\
n_2(u, v) & = \sum_{x \in \ch(u)} {|\Leaves(T_1(x)) \cap
\Leaves(T_2(v))| \choose 2} \cdot |\Leaves(\overline{T_1(u)}) \cap
\Leaves(\overline{T_2(v)})|,\\
n_3(u, v) & = \sum_{x \in \ch(v)} {|\Leaves(T_1(u)) \cap
\Leaves(T_2(x))| \choose 2} \cdot |\Leaves(\overline{T_1(u)}) \cap
\Leaves(\overline{T_2(v)})|, \quad \text{and}
\\n_4(u, v) & = \sum_{x \in \ch(u)} \sum_{y \in \ch(v)} {|\Leaves(T_1(x)) \cap \Leaves(T_2(y))| \choose 2} \cdot |\Leaves(\overline{T_1(u)}) \cap \Leaves(\overline{T_2(v)})|.
\end{align*}

Consider any triplet tree, $ab|c$, counted in $s(u, v)$. It can be verified that $ab|c$ must satisfy the following three conditions: (i) $a, b \in \Leaves(T_1(u)) \cap \Leaves(T_2(v))$ and $c \in \Leaves(\overline{T_1(u)}) \cap \Leaves(\overline{T_2(v)})$, (ii) there does not exist any $x \in \ch(u)$ such that $a, b \in \Leaves(T_1(x))$, and (iii) there does not exist any $x \in \ch(v)$ such that $a, b \in \Leaves(T_2(x))$.
Moreover, observe that any triplet tree $ab|c$ that satisfies these three conditions is counted in $s(u, v)$. %; these three conditions are thus necessary and sufficient.
Therefore, $s(u, v)$ is exactly the number of triplets trees that satisfy all three conditions (i), (ii) and (iii).

The number of triplet trees that satisfy condition (i) is given by $n_1(u, v)$. Some of the triplet trees that satisfy condition (i) may not satisfy conditions (ii) or (iii); these must not be counted in $s(u, v)$.  The value $n_2(u, v)$ is exactly the number of triplet trees that satisfy condition (i) but not condition (ii). Similarly, $n_3(u, v)$ is exactly the number of triplet trees that satisfy condition (i) but not (iii).
Thus, the second and third terms must be subtracted from the first term. However, there may be triplet trees that satisfy condition (i) but neither (ii) nor (iii), and, consequently,  get subtracted in both the second and third terms. In order to adjust for these, the value $n_4(u, v)$ counts exactly those triplet trees that satisfy condition (i) but not (ii) and (iii).
\end{proof}

A summary of our algorithm to compute $|\SameRes(T_1, T_2)|$ appears in Figure~\ref{SharedTriplets}.

%\begin{algorithm}
\begin{figure}

\begin{algorithmic}[1]
\REQUIRE $\SameRes(T_1, T_2)$

%\State Initialize a variable $score \leftarrow 0$.

\FOR{each internal node $u \in \V(T_1) \setminus \rt(T_1)$} \label{alg:line4}

\FOR{each internal node $v \in \V(T_2) \setminus \rt(T_2)$}

%\State $score = score + s(u, v)$
\STATE Compute $s(u, v)$.
\ENDFOR

\ENDFOR

\STATE \textbf{return} the sum of all computed $s(\cdot, \cdot)$.

\end{algorithmic}
\caption{Computing $|\SameRes(T_1,T_2)|$} \label{SharedTriplets}
\end{figure}

\begin{lemma} \label{triplet:lemma7}
Given two rooted phylogenetic trees $T_1$ and $T_2$ on the same $n$
leaves, the value $|\SameRes(T_1, T_2)|$ can be computed in
$O(n^2)$ time.
\end{lemma}

\begin{proof}
By Lemma~\ref{triplet:lemmaMainShared}, the algorithm of Figure~\ref{SharedTriplets} computes the value $|\SameRes(T_1,T_2)|$ correctly. We now analyze its complexity. The running time of the algorithm is
dominated by the complexity of computing the value $s(u, v)$ for each pair of internal nodes $u \in \V(T_1)$ and $v \in \V(T_2)$.
 According to Lemma~\ref{lem:Shared_count}, the value $s(u, v)$ can be computed in $O(|\ch(u)| \cdot |\ch(v)|)$ time. Thus, the total time complexity of the algorithm is $O(\sum_{u \in \V(T_1)} \sum_{v \in \V(T_2)} |\ch(u)| \cdot |\ch(v)|)$, which is $O(n^2)$.
\end{proof}

\subsection{Computing $|\Res_1(T_1, T_2)|$}
%\subsection{Counting triplets resolved in only one tree}

\label{sec:R1_alg}

Next, we describe an $O(n^2)$-time algorithm that computes the cardinality of the set $\Res_1(T_1, T_2)$ of triplets that are resolved only in tree $T_1$. First, we need a definition. Let $X$ be a triplet that is unresolved in $T_2$. Let $v$ be the least common ancestor (lca) of $X$ in $T_2$.
We say that $X$ is \emph{associated} with $v$. Observe that node
$v$ must be internal and unresolved.  Note also that $X$ is
associated with exactly one node in $T_2$.

For any $u \in \V(T_1) \setminus (\rt(T_1) \cup \Leaves(T_1))$ and $v \in \V(T_2) \setminus \Leaves(T_1))$, let $r_1(u, v)$ denote the number of triplets $X$ such that $T_1|X$ is strictly induced by edge $\{u, \pa(u)\}$ in $T_1$, and $X$ is associated with the node $v$ in $T_2$.

The triplets counted in $r_1(u, v)$ must be resolved in $T_1$ but unresolved in $T_2$. Our algorithm computes the value $|\Res_1(T_1, T_2)|$ by
computing, for each $u \in \V(T_1) \setminus (\rt(T_1) \cup \Leaves(T_1))$ and $v \in \V(T_2) \setminus \Leaves(T_2)$, the value $r_1(u, v)$. We claim that the sum of all the computed $r_1(u,v)$'s yields the value $|\Res_1(T_1, T_2)|$.

\begin{lemma} \label{triplet:lemmaMainRes}
Given $T_1$ and $T_2$, we have,

\begin{equation} \label{Lemma:tripleteqnMainRes}
|\Res(T_1, T_2)| =  \sum_{\substack{u \in \V(T_1) \setminus (\rt(T_1) \cup \Leaves(T_1)), \\  v \in \V(T_2) \setminus \Leaves(T_2)}} r_1(u, v).
\end{equation}
\end{lemma}
\begin{proof}
Consider any triplet $X \in \Res_1(T_1, T_2)$. %Without any loss of generality we may assume that $X = \{a, b, c\}$ and $T_1|X = ab|c$.
By Lemma~\ref{lem:strictlyinduced}, there exists exactly one node $u \in \V(T_1)\setminus \rt(T_1)$ such that the edge $\{u, \pa(u)\}$ strictly induces $T_1|X$ in $T_1$. Also observe that there must be exactly one unresolved  node $v \in \V(T_2)$ with which $X$ is associated. Additionally, neither $u$ nor $v$ can be leaf nodes in $T_1$ and $T_2$ respectively.
Thus, $X$ would be counted exactly once in the right-hand side of Equation~\eqref{Lemma:tripleteqnMainRes}; in the value $r_1(u, v)$. Moreover, by the definition of $r_1(u, v)$, any triplet that is counted in the right-hand side of Equation~\eqref{Lemma:tripleteqnMainRes} must belong to the set $\Res_1(T_1, T_2)$. %be resolved, and resolved identically, in $T_1$ and $T_2$.
The lemma follows.
\end{proof}

%A formal description of our algorithm appears in Algorithm~\ref{Alg-R1}.

Given a path $u_1, u_2, \ldots, u_k$, where $k \geq 2$, in tree $T_1$ such that $u_k$ is an internal node and $u_1$ is an ancestor of $u_k$, let $\gamma(u_1, u_k, v)$
denote the number of triplets $X$ such that $T_1|X$ is induced by every edge $\{u_{i-1}, u_i\}$,  for $2 \leq i \leq k$, in $T_1$ and $X$ is associated with node $v$ in $T_2$.

%Our algorithm essentially computes the sum of the values of $r_1(u, v)$ for all candidate $u$ and $v$.
The following lemma shows how the value of $r_1(u, v)$ can be computed by first computing certain $\gamma(\cdot, \cdot, \cdot)$ values.

\begin{lemma} \label{lem:count_gamma}
For any $u \in \V(T_1) \setminus (\rt(T_1) \cup \Leaves(T_1))$ and $v \in \V(T_2) \setminus \Leaves(T_2))$, $$r_1(u, v) = \gamma(\pa(u), u, v) - \sum_{x \in \ch(u)}\gamma(\pa(u), x, v).$$
\end{lemma}
\begin{proof}
%$\Gamma(\pa(u), u, v)$ counts exactly those triplets whose triplet trees are induced by edge $(u, \pa)$ in $T_1$ and that are associated with node $v$ in $T_2$. Similarly, $r_1(u, v)$ counts exactly that subset of the triplets counted in $\Gamma(\pa(u), u, v)$, that are not induced by any other edge in $T_1$.
%
Let $X = \{a, b, c\}$ be a triplet that is counted in $r_1(u, v)$. And, without loss of generality, let $T_1|X = ab|c$. It can be verified that
 $X$ must satisfy the following three conditions: (i) $X$ must be associated with $v$ in $T_2$, (ii) $a, b \in \Leaves(T_1(u))$  and $c \in \Leaves(\overline{T_1(u)})$, and (iii) there must not exist any $x \in \ch(u)$ such that $a, b \in \Leaves(T_1(x))$. Moreover, observe that if there exists a triplet $X = \{a,b,c\}$ that satisfies these three conditions, then $X$ will be counted in $r_1(u, v)$; these three conditions are thus necessary and sufficient.

Now observe that $\gamma(\pa(u), u, v)$ counts exactly those triplets that satisfy conditions (i) and (ii), while  $\sum_{x \in \ch(u)}\gamma(\pa(u), x, v)$ counts exactly those triplets that satisfy conditions (i) and (ii), but not condition (iii). The lemma follows immediately.
\end{proof}

To compute the value of $\gamma(\cdot, \cdot, \cdot)$ efficiently we use the following lemma.

\begin{lemma} \label{triplet:lemma10}

Consider a path $u_1, u_2, \ldots, u_k$, where $k \geq 2$, in tree $T_1$ such that $u_k$ is an internal node and $u_1$ is an ancestor of $u_k$. And let $v \in \V(T_2)$ be an internal unresolved node. Then,
\begin{align*}
\gamma(u_1, u_k, v) & =  n_1(u_1, u_k, v) - n_2(u_1, u_k, v) - n_3(u_1, u_k, v) - n_4(u_1, u_k, v),
\end{align*}
where
\begin{align*}
n_1(u_1, u_k, v) & = {|\Leaves(T_2(v)) \cap \Leaves(T_1(u_k))| \choose 2} \cdot |\Leaves(T_2(v)) \cap \Leaves(\overline{T_1(u_2)})|, \\
n_2(u_1, u_k, v) & = \sum_{x \in \ch(v)} {{|\Leaves(T_2(x)) \cap \Leaves(T_1(u_k))|} \choose 2} \cdot |\Leaves(T_2(x)) \cap
\Leaves(\overline{T_1(u_2)})|, \\
n_3(u_1, u_k, v) & = \sum_{x \in \ch(v)} {{|\Leaves(T_1(u_k)) \cap
\Leaves(T_2(x))|} \choose 2} \cdot \left ( |\Leaves(T_2(v)) \cap \Leaves(\overline{T_1(u_2)})
| - |\Leaves(T_2(x)) \cap \Leaves(\overline{T_1(u_2)})| \right ),
\end{align*}
and
\begin{align*}
n_4(u_1, u_k, v) & = \sum_{x
\in \ch(v)} |\Leaves(T_2(x)) \cap \Leaves(T_1(u_k))| \cdot |\Leaves(T_2(x)) \cap
\Leaves(\overline{T_1(u_2)})| \\
& \hspace*{3cm} \cdot \big ( |\Leaves(T_2(v)) \cap \Leaves(T_1(u_k))| - |\Leaves(T_2(x)) \cap
\Leaves(T_1(u_k))| \big ).
\end{align*}
\end{lemma}

\begin{proof}
Consider those triplets $X$
for which $T_1|X$ is induced by every edge $(u_{i-1}, u_i)$,  for $2 \leq i \leq k$, in $T_1$, and $T_2|X$ is a subtree of $T_2(v)$. Let us call these triplets
\emph{relevant}. Any relevant triplet must have all three leaves from $\Leaves(T_2(v))$, two leaves from $\Leaves(T_1(u_k))$, and the third leaf from $\Leaves(\overline{T_1(u_2)})$. Also note that any triplet that satisfies these three conditions must be relevant.
The number of triplets that satisfy these conditions is exactly $n_1(u_1, u_k, v)$.
%${|\Leaves(T_2(v)) \cap \Leaves(T_1(u_k))| \choose 2} \cdot |\Leaves(T_2(v)) \cap \Leaves(\overline{T_1(u_2)})|$, which is the first term on the RHS of the equation in the lemma.

Any relevant triplet $X$ must belong to one of the following four
categories:
\begin{enumerate}
\item \emph{The lca of $X$ in $T_2$ is not node $v$} : This implies that, in addition to being a relevant triplet, all three leaves of $X$ must belong to the same subtree of $T_2$ rooted at a child of $v$. The number of such triplets is $n_2(u_1, u_k, v)$. %$\sum_{x \in \ch(v)} {|\Leaves(T_2(x)) \cap \Leaves(T_1(u_k))| \choose 2} \cdot |\Leaves(T_2(x)) \cap \Leaves(\overline{T_1(u_2)})|$, which is the second term on the RHS of the equation in the lemma.
\item \emph{The lca of $X$ in $T_2$ is node $v$, $X$ is resolved in $T_2$ and $T_1|X = T_2|X$} : A relevant triplet $X$ satisfies this criterion if and only if there exists a child $x \in \ch(v)$, such that the two leaves of this triplet that belong to $\Leaves(T_1(u_k))$ in tree $T_1$ also occur in $\Leaves(T_2(x))$, and, the third leaf (which occurs in $\Leaves(\overline{T_1(u_2)})|$ in $T_1$) occurs in $\Leaves(T_2(y))$ where $y \in \ch(v)\setminus \{x\}$. The number of such $X$ is equal to $n_3(u_1, u_k, v)$. %$\sum_{x \in \ch(v)} {{|\Leaves(T_1(u_k)) \cap \Leaves(T_2(x))|} \choose 2} \cdot \left ( |\Leaves(T_2(v)) \cap \Leaves(\overline{T_1(u_2)})| - |\Leaves(T_2(x)) \cap \Leaves(\overline{T_1(u_2)})| \right )$, which is the third term on the RHS of the equation in the lemma.
\item \emph{The lca of $X$ in $T_2$ is node $v$, $X$ is resolved in $T_2$, but $T_1|X \neq T_2|X$} : A relevant triplet $X$ satisfies this criterion if and only if there exists a child $x \in \ch(v)$, such that a pair of the leaves of $X$ that occur in $\Leaves(T_1(u_k))$ and $\Leaves(\overline{T_1(u_2)})$ respectively in tree $T_1$ occur in $\Leaves(T_2(x))$ in tree $T_2$, and, the third leaf (which occurs in $\Leaves(T_2(x))$ in $T_1$) occurs in $\Leaves(T_2(y))$ where $y \in \ch(v) \setminus \{x\}$. The number of such $X$ is given by $n_4(u_1, u_k, v)$. %$\sum_{x \in \ch(v)} |\Leaves(T_2(x)) \cap \Leaves(T_1(u_k))| \cdot |\Leaves(T_2(x)) \cap
%\Leaves(\overline{T_1(u_2)})| \cdot \left ( |\Leaves(T_2(v)) \cap \Leaves(T_1(u_k)) | - |\Leaves(T_2(x)) \cap
%\Leaves(T_1(u_k))| \right )$, which is the fourth term on the RHS of the equation in the lemma.
\item \emph{The lca of $X$ in $T_2$ is node $v$, and $X$ is unresolved in $T_2$} :  By definition, the number of relevant triplets that satisfy this criterion is exactly $\gamma(u_1, u_k, v)$.
\end{enumerate}

We have shown that $n_2(u_1, u_k, v)$, $n_3(u_1, u_k, v)$, and $n_4(u_1, u_k, v)$ are exactly the number of relevant triplets belonging to categories 1, 2, and 3 respectively. The lemma follows.
\end{proof}

We should remark that the procedure to compute the value of $\gamma(u_1, u_k, v)$ given in the preceding proof may seem circuitous.   However, we have been unable to find a direct method with an equally good time complexity.

\begin{figure}

\begin{algorithmic}[1]
\REQUIRE $\Res_1(T_1, T_2)$

%\State Initialize a variable $score \leftarrow 0$.

\FOR{each internal node $u \in \V(T_1) \setminus \{\rt(T_1)\}$}

\FOR{each internal unresolved node $v \in \V(T_2)$}

%\State $score = score + r_1(u, v)$
\STATE Compute $r_1(u, v)$.
\ENDFOR

\ENDFOR

\STATE \textbf{return} the sum of all computed $r_1(\cdot, \cdot)$.
\end{algorithmic}
\caption{Computing $|\Res_1(T_1,T_2)|$}
\label{Alg-R1}
\end{figure}

\begin{lemma} \label{triplet:lemma11}
Given two phylogenetic trees $T_1$ and $T_2$ on the same $n$
leaves, the value  $|\Res_1(T_1, T_2)|$ can be computed in
$O(n^2)$ time.
\end{lemma}
\begin{proof}
%We  prove that Algorithm~\ref{Alg-R1} correctly computes the value  $|\Res_1(T_1, T_2)|$ in $O(n^2)$ time.
%\paragraph{Correctness.} Lemma~\ref{triplet:lemmaMainRes} immediately implies that Algorithm~\ref{Alg-R1} computes the value $|\Res_1(T_1, T_2)|$ correctly.
Our algorithm for computing $|\Res_1(T_1, T_2)|$ appears in Figure~\ref{Alg-R1}.  The correctness of the algorithm follows from
Lemma~\ref{triplet:lemmaMainRes}. We now analyze its complexity. For any given candidate nodes $u, v$, Lemma~\ref{triplet:lemma10} shows how to compute $\gamma(\cdot, \cdot, v)$ in $O(|\ch(v)|)$ time, and consequently, by Lemma~\ref{lem:count_gamma}, the value $r_1(u, v)$ can be computed in  $O(|\ch(u)| \cdot |\ch(v)|)$ time. Thus, the total time complexity of the algorithm is $O(\sum_{u \in \V(T_1)} \sum_{v \in \V(T_2)} |\ch(u)| \cdot |\ch(v)|)$, which is $O(n^2)$.
\end{proof}

\section{An approximation algorithm for parametric quartet distance} \label{sec:PQD}

%Our technique for computing the parameterized triplet distance can be extended to compute a 2-approximate value of the parameterized quartet distance (PQD) between two unrooted trees $T_1$ and $T_2$ in $O(n^2)$ time for any $p \geq \frac{1}{2}$. The best known time bound for computing the quartet distance between partially resolved trees is $O(n^2 \min \{d_1, d_2\})$ \cite{ChristiansenMailundPedersenRandersStigStissing06}, where $d_1$ and $d_2$ are the maximum degrees of $T_1$ and $T_2$ respectively. The algorithm in \cite{ChristiansenMailundPedersenRandersStigStissing06} can also be trivially extended to compute the PQD with-in the same time complexity. Our faster $O(n^2)$ algorithm offers a 2-approximate solution when an exact value of the PQD is not required.

We now consider the problem of computing the parametric quartet
distance (PQD) between two unrooted trees. Our main result is an $O(n^2)$-time
2-approximate algorithm for PQD.

Our approach is similar to the one for computing the parametric
triplet distance. Observe that Proposition~\ref{prop:triplet_obs} and,
thus, Equation~\eqref{eqn:param_dist_2} hold even when the unit of
distance is quartets instead of triplets.
Christiansen et al.~\cite{ChristiansenMailundPedersenRandersStigStissing06} show
how to compute the values $|\SameRes(T_1, T_2)|$, $|R(T_1)|$,
$|U(T_1)|$, and $|U(T_2)|$ within $O(n^2)$ time. %An analogue of the preprocessing step also exists (simply root $T_1$ and $T_2$ arbitrarily at any internal node, and proceed as in the rooted case).
In Section~\ref{subsec:R1:approx} we show how to compute,
in $O(n^2)$ time, a value $y$ such that $|\Res_1(T_1, T_2)| \leq y
\leq 2|\Res_1(T_1, T_2)|$. Now, let us substitute the values of
$|R(T_1)|$, $|U(T_1)|$, $|U(T_2)|$ and $|\SameRes(T_1, T_2)|$ into
Equation~\eqref{eqn:param_dist_2}, and use the value of $y$ instead of
$|\Res_1(T_1, T_2)|$.  Assuming $p \geq 1/2$, it can be seen that
the result is a 2-approximation to $d^{(p)}(T_1, T_2)$.

To summarize, we have the following result.

\begin{theorem} \label{quartet:theorem2}
Given two unrooted phylogenetic trees $T_1$ and $T_2$ on the same
$n$ leaves, and a parameter $p \geq 1/2$, a value $x$ such that
$d^{(p)}(T_1, T_2) \leq x \leq 2 \cdot d^{(p)}(T_1, T_2)$ can be
computed in $O(n^2)$ time.
\end{theorem}

We note that the $(2p -1) \cdot |\Res_1(T_1,T_2)|$ term in
Equation~\eqref{eqn:param_dist_2} vanishes when $p=\frac{1}{2}$.  In
this case, we do not even need to compute $|\Res_1(T_1, T_2)|$ to get the \emph{exact} value of $d^{(p)}(T_1, T_2)$.

\subsection{Computing a 2-approximate value of $|\Res_1(T_1, T_2)|$} \label{subsec:R1:approx}

For any node $u$ in $T$, let  $\adj(u)$ denote the set of nodes that are adjacent to $u$.
For the purposes of describing our algorithm, it is useful to view each (undirected) edge $\{u, v\} \in \E(T)$ as two directed edges $(u, v)$ and $(v, u)$.
%(In reality, of course, our trees do not contain any directed edges.)
Let $\overrightarrow{\E}(T)$ denote the set of directed edges in tree $T$.

To achieve the claimed time complexity, our algorithm relies on a preprocessing step which computes and stores, for each pair of directed edges $(u_1, v_1) \in \overrightarrow{\E}(T_1)$ and $(u_2, v_2) \in \overrightarrow{\E}(T_2)$, the quantity $|\Leaves(T_1(u_1, v_1)) \cap
\Leaves(T_2(u_2, v_2))|$. This can be accomplished in $O(n^2)$ by arbitrarily rooting $T_1$ and $T_2$ at any internal node and proceeding as in the preprocessing step for the triplet distance case (see Section~\ref{trip:precomp}).

  %These values are stored in a table so that any value can be accessed in $O(1)$ time by subsequent steps of the PTD algorithm.

%The notion of resolved quartets induced by edges is a direct extension of the corresponding notion for triplets.
Consider any two leaves $a, b$ from $\Leaves(T(u, v))$ and any two leaves
$c, d$ from $\Leaves(T(v, u))$. Then, the quartet $\{a, b,
c, d\}$ must appear resolved as $ab|cd$ in $T$; we say that the quartet tree $ab|cd$ is \emph{induced} by the edge $(u, v)$. Note that
the same resolved quartet tree may be induced by multiple edges in $T$.
Additionally, if $x \in u_1$ and $y \in u_2$ for some $u_1, u_2 \in \adj(u)\setminus \{v\}$ such that $u_1 \neq u_2$, then we say that the quartet tree $ab|cd$ is \emph{strictly induced} by the directed edge $(u, v)$.

Consider a quartet
$\{a,b,c,d\}$. Then, the corresponding quartet tree
is unresolved in $T$ if and only if there exists exactly one node
$w$ such that the paths from $w$ to $a$, $w$ to $b$, $w$ to $c$,
and $w$ to $d$ do not share any edges. We say that quartet
$\{a,b,c,d\}$ is \emph{associated} with node $w$ in $T$. Thus,
each unresolved quartet tree from $T$  is associated with exactly
one node in $T$.

For any directed edge $(u, v) \in \overrightarrow{\E}(T_1)$ and  $w \in \V(T_2) \setminus \Leaves(T_1)$, let $r_1((u, v), w)$ denote the number of quartets $X$ such that $T_1|X$ is strictly induced by the directed edge $(u, v)$ in $T_1$, and $X$ is associated with the node $w$ in $T_2$.
The quartets counted in $r_1((u, v), w)$ must be resolved in $T_1$ but unresolved in $T_2$. We have the following result.

\begin{lemma} \label{quartet:lemmaMain1}
Given $T_1$ and $T_2$, we have
$$ 2 \cdot |\Res_1(T_1, T_2)| =  \sum_{\substack{(u, v) \in \overrightarrow{\E}(T_1), \\ w \in \V(T_2) \setminus \Leaves(T_2)}} r_1((u, v), w).$$
\end{lemma}
\begin{proof}
Let $X = \{a, b, c, d\}$ be any quartet in $|\Res_1(T_1, T_2)|$. Without loss of generality, assume that $T_1|X = ab|cd$, and that $X$ is associated with node $w \in V(T_2)\setminus \Leaves(T_2)$. Since $X$ appears resolved in $T_1$, $\overrightarrow{\E}(T_1)$ must have exactly two directed edges, say $(u_1, v_1)$ and $(u_2, v_2)$, which strictly induce $ab|cd$. Thus, $X$ is counted in exactly two of the $r_1(\cdot , \cdot)$'s, namely,
$r_1((u_1, v_1), w)$, and $r_1((u_2, v_2), w)$. The lemma follows.
\end{proof}

Thus, we can compute $|\Res_1(T_1, T_2)|$ by computing all the $O(n^2)$ possible $r_1((u, v), w)$'s. However, doing so seems to require at least $\Theta(n^2\cdot d)$ time, where $d$ is the degree of $T_1$. Instead, our algorithm computes a 2-approximate value of $|\Res_1(T_1, T_2)|$ in $O(n^2)$ time by relying on the next lemma. %A

\begin{lemma} \label{quartet:lemmaMain2}
Given $T_1$ and $T_2$, let $T'_1$ denote the rooted tree obtained from  $T_1$ by designating any internal node in $V(T_1)$ as the root. Then,
$$|\Res_1(T_1, T_2)| \leq  \sum_{\substack{u \in \V(T'_1) \setminus (\rt(T'_1) \cup \Leaves(T'_1)), \\ w \in \V(T_2) \setminus \Leaves(T_2)}} r_1((u, \pa(u)), w) \leq 2 \cdot |\Res_1(T_1, T_2)|.$$
\end{lemma}
\begin{proof}
First, observe that if $u \in \Leaves(T'_1)$ and $w \in \V(T_2) \setminus \Leaves(T_2)$, then $r_1((u, \pa(u)), w) = 0$. Therefore, we must have $$\sum_{\substack{u \in \V(T'_1) \setminus (\rt(T'_1) \cup \Leaves(T'_1)), \\ w \in \V(T_2) \setminus \Leaves(T_2)}} r_1((u, \pa(u)), w) = \sum_{\substack{u \in \V(T'_1) \setminus \rt(T'_1), \\ w \in \V(T_2) \setminus \Leaves(T_2)}} r_1((u, \pa(u)), w).$$

Second, observe that $\E(T_1) = \E(T'_1)$ and, therefore, by Lemma~\ref{quartet:lemmaMain1}, we must have $$\sum_{\substack{u \in \V(T'_1) \setminus \rt(T'_1), \\ w \in \V(T_2) \setminus \Leaves(T_2)}} r_1((u, \pa(u)), w) \leq 2 \cdot |\Res_1(T_1, T_2)|.$$ This proves the second inequality in the lemma.

To complete the proof, we now prove the first inequality.
Let $X = \{a, b, c, d \}$ be any quartet in $|\Res_1(T_1, T_2)|$, and, without loss of generality, assume that $T_1|X = ab|cd$, and that $X$ is associated with node $w \in V(T_2)\setminus \Leaves(T_2)$. Since $X$ appears resolved in $T_1$, $\overrightarrow{\E}(T_1)$ must have exactly two directed edges, say $(u_1, v_1)$ and $(u_2, v_2)$, which strictly induce $ab|cd$. Consider the edge $\{u_1, v_1\} \in \E(T'_1)$. There are two possible cases: Either $v_1 = \pa(u_1)$, or  $u_1 = \pa(v_1)$. If $v_1 = \pa(u_1)$ then the quartet $X$ will be counted in the value $r_1((u_1, \pa(u_1)), w)$. Otherwise, if $u_1 = \pa(v_1)$, then $u_1, v_1, v_2, u_2$ must appear on a same root-to-leaf path in $T'_1$. Consequently, we must have  $v_2 = \pa(u_2)$ and the quartet $X$ would be counted in the value $r_1((u_2, \pa(u_1)), w)$. Thus, we must have $|\Res_1(T_1, T_2)| \leq  \sum_{u \in \V(T'_1) \setminus \rt(T'_1)} \sum_{w \in \V(T_2) \setminus \Leaves(T_2)} r_1((u, \pa(u)), w)$. The lemma follows.
\end{proof}

Thus, the idea for efficiently computing a 2-approximate value of $|\Res_1(T_1, T_2)|$ is to first root $T_1$ arbitrarily at any internal node and then compute the value $r_1((u, \pa(u)), w)$ for each non-root node $u \in V(T_1)$ and each  $w \in \V(T_2) \setminus \Leaves(T_1)$.

We now direct our attention to the problem of efficiently computing all the required values $r_1(\cdot, \cdot)$.
Given a path $u_1, u_2, \ldots, u_k$ in $T_1$, where $k \geq 2$, let $\gamma(u_1, u_k, w)$ denote the number of quartets $X$ such that $T_1|X$ is induced in $T_1$ by every edge $(u_{i-1}, u_i)$, $2 \leq i \leq k$, and $X$ is associated with node $w$ in $T_2$.

The following lemma is analogous to Lemma~\ref{lem:count_gamma}, and shows how the value $r_1(\cdot, \cdot)$ can be computed by first computing certain $\gamma(\cdot, \cdot, \cdot)$ values.

\begin{lemma} \label{quartetlem:count_gamma}
Let $(u, v) \in \E(T_1)$, and $w \in \V(T_2) \setminus \Leaves(T_2))$, then, $$r_1((u, v), w) = \gamma(u, v, w) - \sum_{x \in \adj(u) \setminus \{v\}}\gamma(x, v, w).$$
\end{lemma}
\begin{proof}
Let $X = \{a, b, c, d\}$ be a quartet that is counted in $r_1((u, v), w)$.  Without loss of generality, let $T_1|X = ab|cd$ such that $a, b \in \Leaves(T_1(u, v))$. It can be verified that
 $X$ must satisfy the following three conditions: (i) $X$ must be associated with node $w$ in $T_2$, (ii) $a, b \in \Leaves(T_1(u, v))$  and $c, d \in \Leaves(T_1(v, u))$, and (iii) there must not exist any $x \in \adj(u) \setminus \{v\}$ such that $a, b \in \Leaves(T_1(x, u))$. Moreover, observe that if there exists a quartet $X = \{a,b,c, d\}$ that satisfies these three conditions, then $X$ will be counted in $r_1((u, v), w)$; these three conditions are thus necessary and sufficient.

Now observe that $\gamma(u, v, v)$ counts exactly all those quartets that satisfy conditions (i) and (ii), while  $\sum_{x \in \ch(u)}\gamma(\pa(u), x, v)$ counts exactly all those quartets that satisfy conditions (i) and (ii), but not condition (iii). The lemma follows.
\end{proof}

%To efficiently compute the value $\gamma(\cdot, \cdot, \cdot)$, in turn, we rely on Lemma~\ref{quartet:lemma10}.
To state our next results we need the following notation.
Given phylogenetic trees $T_1$ and $T_2$, consider a path $u_1,
u_2, \ldots, u_k$ where $k \geq 2$, in tree  $T_1$, and an
internal node $w \in \V(T_2)$ of degree at least 4. Let $P =
\Leaves(T_1(u_1, u_2))$, $Q = \Leaves(T_1(u_k, u_{k-1}))$ and let $x_1, \ldots, x_{|\adj(w)|}$ denote
the neighbors of $w$. %Let $\Gamma$ be the set of quartets that
%are induced by every edge $(u_{i-1}, u_i)$, $2 \leq i \leq k$, in
%$T_1$ and associated with $w$ in $T_2$.
Consider the quartets that are induced by every edge $(u_{i-1}, u_i)$, $2 \leq i \leq k$, in $T_1$: Let us call these quartets \emph{relevant}. Observe that a quartet is relevant if and only if it contains exactly two leaves from $P$ and two leaves from $Q$. Let

%We define five values, denoted $n_i(u_1, u_k, w)$ for

\begin{enumerate}

 \item $n_1(u_1, u_k, w)$ denote the number of relevant quartets $X$ for which there exists a neighbor $x$ of $w$ in tree $T_2$, such that $X$ is completely contained
in $T_2(x, w)$,
\item  $n_2(u_1, u_k, w)$ denote the number of relevant quartets $X$ for which there exist two neighbors $x, y$ of
$w$ in tree $T_2$, such that $T_2(x, w)$ contains three
leaves from $X$ and $T_2(y, w)$ contains the
other leaf,
\item $n_3(u_1, u_k, w)$ denote the number of relevant quartets $X$ for which there exist two neighbors $x, y$ of $w$ in tree
$T_2$, such that $T_2(x, w)$ contains two leaves from $X$
and $T_2(y, w)$ contains the other two leaves, and
\item $n_4(u_1, u_k, w)$ denote the number of relevant quartets $X$ for which there exist three neighbors $x, y, z$ of $w$ in tree
$T_2$, such that $T_2(x, w)$ contains two leaves from
$X$, $T_2(y, w)$ contains one leaf from $X$, and $T_2 (z, w)$ contains the remaining leaf.
\end{enumerate}
Then, we must have the following.

%------------------ Start of Part 8 Modification -----------------
%To prove the correctness of Algorithm~\ref{Alg-Quartet-Approx-R1}, we need the following auxiliary result.

\begin{lemma} \label{quartet:lemma10}
\begin{equation}
\gamma(u_1, u_k, w) = {|P| \choose 2} \cdot {|Q| \choose 2} - n_1(u_1, u_k, w) - n_2(u_1, u_k, w) - n_3(u_1, u_k, w) - n_4(u_1, u_k, w).
\end{equation}
\end{lemma}
\begin{proof}
The term ${|P| \choose 2} \cdot {|Q| \choose 2}$ is the number of
relevant quartets. Furthermore, each relevant quartet
must occur in tree $T_2$ in exactly one of the five configurations captured by the terms
$n_1(u_1, u_k, w)$, $n_2(u_1, u_k, w)$, $n_3(u_1, u_k, w)$, $n_4(u_1, u_k, w)$, and $\gamma(u_1, u_k, w)$. The lemma follows.
\end{proof}

The following four lemmas show that the values of $n_1(u_1, u_k, w)$, $n_2(u_1, u_k, w)$, $n_3(u_1, u_k, w)$, and $n_4(u_1, u_k, w)$ can be computed in $O(|\adj(w)|)$ time. The proofs of these lemmas all follow the same approach:  In each case, we show that the required value can be expressed as a sum of $O(|\adj(w)|)$ quantities, every one of which can be computed in $O(1)$ time based on the values computed in the pre-processing step.

\begin{lemma} \label{quartet:lemma:10:1}
The value $n_1(u_1, u_k, w)$ can be computed in $O(|\adj(w)|)$ time.
%Following the notation established in Lemma~\ref{quartet:lemma10}, we must have,
\end{lemma}
\begin{proof}
We will show that \begin{equation} \label{quartet:eqn2}
n_1(u_1, u_k, w) = \sum_{i=1}^{|\adj(w)|} {{|\Leaves(T_2(x_i, w)) \cap P|} \choose 2} \cdot {{|\Leaves(T_2(x_i, w)) \cap Q|} \choose 2}.
\end{equation}

The right hand side of Equation~\eqref{quartet:eqn2} counts all those quartets that are completely contained in $\Leaves(T_2(x, w))$ for some $x \in \adj(w)$ and that have two elements from $P$ and two from $Q$. These are exactly the quartets that must be counted in $n_1(u_1, u_k, w)$.
\end{proof}

\begin{lemma} \label{quartet:lemma:10:2}
The value $n_2(u_1, u_k, w)$ can be computed in $O(|\adj(w)|)$ time.
\end{lemma}
\begin{proof}
We will show that
\begin{align} \label{quartet:eqn3}
%\hspace*{-0.6cm}
n_2(u_1, u_k, w) & =  \sum_{i=1}^{|\adj(w)|}{{|\Leaves(T_2(x_i, w)) \cap P|} \choose 2} \cdot {|\Leaves(T_2(x_i, w)) \cap Q|} \cdot {|\Leaves(T_2(w, x_i)) \cap Q|} \nonumber  \\ & +  \sum_{i=1}^{|\adj(w)|}{{|\Leaves(T_2(x_i, w)) \cap Q|} \choose 2} \cdot {|\Leaves(T_2(x_i, w)) \cap P|} \cdot {|\Leaves(T_2(w, x_i)) \cap P|}.
\end{align}

The quartets $X$ counted in $n_2(u_1, u_k, w)$ are exactly those for which there exist two neighbors $x, y$ of $w$ such that either (i) $X \cap \Leaves(T_2(x, w))$ contains two leaves from $P$ and one from $Q$, and $X \cap \Leaves(T_2(y, w))$ contains a leaf from $Q$ or (ii) $X \cap \Leaves(T_2(x, w))$ contains two leaves from $Q$ and one from $P$, and $X \cap \Leaves(T_2(y, w))$ contains a leaf from $P$. The first term on the right hand side of Equation~\eqref{quartet:eqn3} is exactly the number of quartets that satisfy condition (i), and the second term on the right hand side is exactly the number of quartets satisfying condition (ii).
\end{proof}

\begin{lemma} \label{quartet:lemma:10:3}
The value $n_3(u_1, u_k, w)$ can be computed in $O(|\adj(w)|)$ time.
%Following the notation established in Lemma~\ref{quartet:lemma10}, we must have,
\end{lemma}
\begin{proof}
We will show that
\begin{align} \label{quartet:eqn4}
n_3(u_1, u_k, w)   & =   \sum_{i=1}^{|\adj(w)|}\left \{ \alpha - {{|\Leaves(T_2(x_i, w)) \cap P|} \choose 2}\right \} \cdot {{|\Leaves(T_2(x_i, w)) \cap Q|} \choose 2}  \nonumber \\ & +   \frac{1}{2} \sum_{i=1}^{|\adj(w)|} \left \{ \beta - {{|\Leaves(T_2(x_i, w)) \cap P|} \cdot {|\Leaves(T_2(x_i, w)) \cap Q|} }\right \} \cdot {|\Leaves(T_2(x_i, w)) \cap P|} \nonumber \\ &   &  \hspace*{6cm} \cdot {|\Leaves(T_2(x_i, w)) \cap Q|}.
\end{align}
 Where
 \begin{equation}
 \alpha = \sum_{i=1}^{|\adj(w)|}{{|\Leaves(T_2(x_i, w)) \cap P|} \choose 2},
 \end{equation}
 \begin{equation}
 \beta = \sum_{i=1}^{|\adj(w)|}{|\Leaves(T_2(x_i, w)) \cap P|} \cdot {|\Leaves(T_2(x_i, w)) \cap Q|}.
\end{equation}

The quartets $X$ counted in $n_3(u_1, u_k, w)$ are exactly those quartets for which there exist two neighbors $x, y$ of $w$ such that either (i) $X \cap \Leaves(T_2(x, w))$ contains two leaves from $P$, and $T_2(y, w)$ contains two leaves from $Q$, or (ii) $X \cap \Leaves(T_2(x, w))$ and $X \cap \Leaves(T_2(y, w))$ both contain one leaf each from $P$ and $Q$. The first term on the right hand side of Equation~\eqref{quartet:eqn4} is exactly the number of quartets that satisfy condition (i). The sum in the second term on the right hand side counts the quartets satisfying condition (ii) exactly twice each (due to the symmetry between $x$ and $y$ in condition (ii)). This explains the $\frac{1}{2}$ multiplicative factor.
\end{proof}

\begin{lemma}\label{quartet:lemma:10:4}
The value $n_4(u_1, u_k, w)$ can be computed in $O(|\adj(w)|)$ time.
\end{lemma}
\begin{proof}
We will show that
\begin{align} \label{quartet:eqn5}
n_4(u_1, u_k, w)  & = \sum_{i=1}^{|\adj(w)|}{{|\Leaves(T_2(x_i, w)) \cap P|} \choose 2} \cdot {{|\Leaves(T_2(w, x_i)) \cap Q|} \choose 2}   \nonumber \\ & +  \sum_{i=1}^{|\adj(w)|}{{|\Leaves(T_2(x_i, w)) \cap Q|} \choose 2} \cdot {{|\Leaves(T_2(w, x_i)) \cap P|} \choose 2}  \nonumber \\ & +  \sum_{i=1}^{|\adj(w)|}{|\Leaves(T_2(x_i, w)) \cap P|} \cdot {|\Leaves(T_2(x_i, w)) \cap Q|} \cdot {|\Leaves(T_2(w, x_i)) \cap P|} \cdot {|\Leaves(T_2(w, x§_i)) \cap Q|} \nonumber \\ & - 2 \cdot n_3(u_1, u_k, w).
\end{align}

The quartets $X$ counted in $n_4(u_1, u_k, w)$ are exactly those quartets for which there exist three neighbors $x, y, z$ of $w$ such that either (i) $X \cap \Leaves(T_2(x, w))$ contains two leaves from $P$, and $T_2(y, w)$ and $T_2(z, w)$ each contain a leaf from $Q$, or (ii) $X \cap \Leaves(T_2(x, w))$ contains two leaves from $Q$, and $X \cap \Leaves(T_2(y, w))$ and $X \cap \Leaves(T_2(z, w))$ each contain a leaf from $P$, or (iii) $X \cap \Leaves(T_2(x, w))$ contains a leaf from $P$ and a leaf from $Q$, $X \cap \Leaves(T_2(y, w))$ contains a leaf from $P$, and $X \cap \Leaves(T_2(z, w))$ contains a leaf from $Q$.

The first term on the right hand side of Equation~\eqref{quartet:eqn5} counts all the quartets that satisfy condition (i), and, in addition, all the quartets that satisfy condition (i) from the proof of Lemma~\ref{quartet:lemma:10:3}. Similarly, the second term on the right hand side counts the quartets that satisfy condition (ii), along with all the quartets that satisfy condition (i) from the proof of Lemma~\ref{quartet:lemma:10:3}. The third term on the right hand side counts those quartets that satisfy condition (iii), and also counts, exactly twice each (again due to symmetry),  those that satisfy condition (ii) from the proof of Lemma~\ref{quartet:lemma:10:3}. Thus, by adding the first three terms on the right hand side of Equation~\eqref{quartet:eqn5}, we obtain the value $n_4(u_1, u_k, w) + 2 \cdot n_3(u_1, u_k, w)$.
\end{proof}

\begin{figure}
\begin{algorithmic}[1]
\REQUIRE Approx-$\Res_1(T_1, T_2)$

\STATE Convert the unrooted tree $T_1$ into a rooted one by rooting it at any internal node.

\FOR{each internal node $u \in \V(T_1) \setminus \rt(T_1)$}

\FOR{each internal unresolved node $w \in \V(T_2)$}

%\State $score = score + r_1(u, v)$
\STATE Compute $r_1((u, \pa(u)), w)$.
\ENDFOR

\ENDFOR

\STATE \textbf{return} the sum of all computed $r_1(\cdot, \cdot)$.
\end{algorithmic}
\caption{Computing a 2-approximation to $|\Res_1(T_1, T_2)|$}\label{Alg-Quartet-Approx-R1}
\end{figure}

\begin{lemma} \label{quartet:lemma11}
Given two unrooted phylogenetic trees $T_1$ and $T_2$ on the same
size $n$ leaf set, a value $y$ such that  $|\Res_1(T_1, T_2)| \leq y \leq 2
\cdot |\Res_1(T_1, T_2)|$ can be computed in $O(n^2)$ time.
\end{lemma}
\begin{proof}

Our algorithm to compute a 2-approximate value of  $|\Res_1(T_1, T_2)|$ is summarized in Figure~\ref{Alg-Quartet-Approx-R1}.
Lemma~\ref{quartet:lemmaMain2} immediately implies that the algorithm computes a value between $|\Res_1(T_1, T_2)|$ and $2
\cdot |\Res_1(T_1, T_2)|$.

We now analyze the time complexity of our algorithm.
By Lemmas~\ref{quartet:lemma:10:1}, \ref{quartet:lemma:10:2}, \ref{quartet:lemma:10:3}, and \ref{quartet:lemma:10:4}, the values $n_1(u_1, u_k, w)$, $n_2(u_1, u_k, w)$, $n_3(u_1, u_k, w)$, and $n_4(u_1, u_k, w)$ can all be computed within $O(|\adj(w)|)$ time. Hence, by Lemma~\ref{quartet:lemma10}, the value of any $\gamma(\cdot, \cdot, w)$ can be computed in $O(|\adj(w)|)$ time. Lemma~\ref{quartetlem:count_gamma} now implies that, for any given $(u, v) \in \overrightarrow{\E}(T_1)$ and $w \in \V(T_2) \setminus \Leaves(T_2)$, the value $r_1((u,v), w)$ can be computed within $O(|\adj(u)|\cdot |\adj(w)|)$ time.
Thus, the total time complexity of the algorithm is $O(\sum_{u \in \V(T'_1)} \sum_{w \in \V(T_2)} |\ch(u)| \cdot |\adj(w)|)$, which is $O(n^2)$.
\end{proof}

\section{Discussion}

We have defined and analyzed distance measures for rooted and unrooted phylogenies that account for partially-resolved nodes.  A number of problems remain.  First, there is the question of  determining whether there exists a polynomial-time algorithm for computing the median tree with respect to parametric triplet and quartet distances.  We conjecture that this problem is NP-hard.  Also open is the question of whether the Hausdorff distance between partially-resolved trees is NP-hard.  Finally, many (if not most) applications require the comparison of trees that do not have the same set of taxa.  It would be interesting to investigate whether any of our distance measures can be extended to this setting.

%\bibliographystyle{abbrv}

%\bibliography{mybib,phylogenies,ranking}

\end{document}